\documentclass[journal]{IEEEtran}
\usepackage[utf8]{inputenc}
 \pdfoutput=1
\usepackage{graphicx}
\usepackage{multirow}
\usepackage{algorithm}
\usepackage{algpseudocode}
\usepackage{amsmath}
\usepackage{multicol}
\usepackage[export]{adjustbox}
\usepackage{float}
\usepackage{amssymb}
\usepackage{lipsum}
\usepackage{booktabs,siunitx}
\usepackage{cite}
\usepackage{comment}
\usepackage{color}
\usepackage{cases}
\usepackage{seqsplit}
\usepackage{mwe}
\usepackage{multirow}
\usepackage{epstopdf}
\usepackage[caption=false]{subfig}
\usepackage{caption}
\usepackage{adjustbox}
\usepackage{amsthm}
\usepackage[american]{babel}
\makeatletter
\newcommand*{\rom}[1]{\expandafter\@slowromancap\romannumeral #1@}
\makeatother
\definecolor{dukeblue}{rgb}{0.0, 0.0, 0.61}
\newcommand{\sv}{\color{black}}

\newtheorem{theorem}{Theorem}[section]

\newtheorem{proposition}[theorem]{Proposition}

\begin{document}
\title{Two-Stage Stochastic Choice Modeling Approach for Electric Vehicle Charging Station Network Design in Urban Communities }

\author{Seyed Sajjad Fazeli, Saravanan Venkatachalam, Ratna Babu Chinnam, Alper Murat
\thanks{S.S. Fazeli, S. Venkatachalam, R.B. Chinnam, and A. Murat are members of the Department of Industrial and Systems Engineering at Wayne State University, Detroit, Michigan (e-mail: sajjad.fazeli@wayne.edu; saravanan.v@wayne.edu;
	ratna.chinnam@wayne.edu; alper@eng.wayne.edu), Corresponding author: S.Venkatachalam}}
\maketitle
\begin{abstract}
Electric vehicles (EVs) provide a cleaner alternative that not only reduces greenhouse gas emissions but also improves air quality and reduces noise pollution. The consumer market for electrical vehicles is growing very rapidly. Designing a network with adequate capacity and types of public charging stations is a challenge that needs to be addressed to support the current trend in the EV market. In this research, we propose a choice modeling approach embedded in a two-stage stochastic programming model to determine the optimal layout and types of EV supply equipment for a community while considering randomness in demand and drivers' behaviors. Some of the key random data parameters considered in this study are: EV's dwell time at parking {\sv location}, battery's state of charge, distance from home, willingness to walk, drivers' arrival patterns, and traffic on weekdays and weekends. The two-stage model uses the sample average approximation method, which asymptotically converges to an optimal solution. To address the computational challenges for large-scale instances, we propose an outer approximation decomposition algorithm. We conduct extensive computational experiments to quantify the efficacy of the proposed approach. In addition, we present the results and a sensitivity analysis for a case study based on publicly available data sources.  
\end{abstract}

\section{Introduction}
One of the most promising approaches to alleviating vehicle emissions and satisfying climate targets is the deployment of electric vehicles \cite{wright2005climate}. Lower maintenance costs, lower ownership costs, noise reduction, and charging at home and work and around the community are some of the additional advantages of using EVs. Vehicle purchasing subsidies, public electric charging availability, and carpool lane access are the three most substantial benefits offered to EV consumers \cite{lutsey2015assessment}. In response to the government's promotion of vehicle electrification objectives, the world's major automobile companies are striving to produce affordable EVs for environmentally conscious consumers \cite{kley2011new}. Every year, automotive companies  around the world introduce various new models of EVs (e.g., hybrid vehicles, plug-in hybrid vehicles, and pure battery electric vehicles (BEVs)). The U.S. is one of the growing markets for EVs. However, half of the U.S. population live in areas with fewer than 90 charging infrastructures per million people, which is 70\% below the estimated benchmarks\cite{slowik2018continued}. By the end of 2025, there should be about a 20\% growth in deployment of charging infrastructures per year to support more than three million expected EVs \cite{nicholas2019quantifying}. Therefore, designing a cost-efficient charging network with broad access is critical for supporting the current flourishing trend in the EV market. 
\newline
 \indent
Installation of a public charging station costs at least \$5,000 to \$15,000 \cite{Global}.
Electric vehicle charging stations (EVCSs) can be equipped with different types of chargers that differ in power, installation cost, and charging price. Broadly, EV supply equipment (EVSE) can be classified into level 1, level 2, and level 3,  based on the power supply. Level 1, which  is known as home charging, has a  1.9kW electric power supply and requires between 8 and  30 hours to fully charge an EV's battery, depending on its size. Level 2, known as semi-rapid charging, has a 6.6 kW power supply and a charging time between 4 and 8 hours. Level 3, known as fast charging EVSE, has a 50kW power supply and a charging time of less than 30 minutes; this is considered to be the most expensive charger.
Given the availability of chargers with different capabilities and prices, it is worth considering EV users' choices of charger levels when establishing an optimal EVCS network. 
\newline
\indent
A study of EV users' charging behaviors, especially their preferences in charging levels and locations, can help increase the accessibility of charging stations for EV users, and this can lead to widespread EV adoption. Also, since establishing an EVCS network is a strategic decision, considering the randomness in demand and analyzing EV users' travel patterns, charging behaviors, and infrastructure utilization will help in designing a charging station network that provides better access \cite{xu2017joint}.
Increasing the overall utilization of charging stations can potentially increase investment opportunities for EVCS providers and automobile makers. The analysis in \cite{nicholas2019quantifying} indicates that EVCS providers can make low-risk and high-utilization investment decisions by expanding charging infrastructures in such a way that the designs of charging outlet networks are matched to the complex driver charging patterns. Different types of users (residential, visitors, employees, fleet users) have different charging needs, as well as different dwell times, frequencies of charging and states of charge (SOCs). Since an EV can be recharged at home, at public charging stations, or at private working places, a wide range of consumers demand several different power supply options. Furthermore, charging prices can significantly affect EV owners' choices. The importance of this last factor can vary depending on people's socioeconomic characteristics. The authors of \cite{wen2016modeling} showed that EV owners are less likely to use charging stations when the charging costs are higher or when their battery has a sufficient driving range for reaching the next charging opportunity.
 \newline
 \indent Many of the existing studies on the charging facility location problem are based on the assumption that charging service demands are deterministic. However, the real demand is affected by various sources of uncertainty, such as the day of the week, the time of day, the purpose of the trip, the location of the final destination, and the driver's willingness to walk. Thus, there may be significant differences between the optimal solutions for deterministic and stochastic models. Stochastic programming is a modelling approach for making decisions under uncertainty. Discrete choice analysis has also proven to be a useful strategy for analyzing and predicting EV drivers' decisions regarding their choices of location and chargers. In this study, considering the uncertainties in EV users' demand, we propose a two-stage stochastic programming model with an embedded choice model representing EV drivers' choices of chargers for designing an optimal network of {\sv charging stations} for a community. Since two-stage stochastic programming models often require a large number of scenarios for good approximations of the expectation function, we use the sample average approximation (SAA) method, a Monte Carlo simulation-based sampling technique. Another challenge for the two-stage stochastic programming approach is the computational burden arising from second-stage scenarios, so we use {\sv a} L-shaped decomposition algorithm with single- and multi-cut variants to solve the model efficiently. Finally, we evaluate our proposed two-stage model and our approach to its solution with a case study based on data representing the midtown area of Detroit, Michigan, in the U.S.
\newline
\indent The contributions of this study include the following: (1) we formulate a two-stage stochastic programming model with an embedded choice model for locating charging facilities, and we determine the types of {\sv chargers} to be installed in these facilities based on EV drivers' choices and behaviors and other random parameters; (2) we include various uncertainties in the model, such as EV demand flows, EV drivers' charging patterns, SOCs, arrival and departure times, the purpose of arrivals in the community, and preferred walking distances; (3) we develop an outer-linearization-based decomposition algorithm and conduct extensive computational experiments with multiple variations to demonstrate the efficacy of our algorithm; and (4) we conduct a case study using data representing the midtown area of Detroit, Michigan, in the U.S. and provide post-analysis insights for improving accessibility and transportation choices based on our proposed framework. {\sv In addition, we conducted a data-driven simulation where the proposed method is compared to two other configurations from the literature.} 
\newline
\indent
The remainder of this paper is organized as follows: Section \ref{lit} reviews the related literature. Section \ref{Pre} describes the various sources of uncertainty that we consider in the demand generation process as well as our construction of the utility function for the choice model. Section \ref{formula} provides a mathematical formulation of the problem along with a subsequent reformulation. Section \ref{Method} introduces the solution methodologies that we implemented to solve large-scale instances. Section \ref{Case} presents the case study, computational experiments, {\sv data-driven simulation} and various insights from our sensitivity analysis. Finally, Section \ref{Con} provides concluding remarks. 

\section{Literature Review}\label{lit}
In this section, we first review the literature related to deterministic and stochastic approaches for the EV charging location problem. Then we provide details about choice models for the behaviors of the EV drivers.\\
\indent A majority of the studies in the literature on the EV charging location problem consider deterministic models. A capacitated refueling location model with limited traffic flow was introduced in \cite{upchurch2009model} to maximize the vehicle miles traveled by alternative-fuel vehicles.
A reformulation of the flow-refueling location model was proposed in \cite{mirhassani2012flexible} to decrease the computational effort needed to  solve large-scale set covering and the maximum coverage problem.
The research in \cite{he2013optimal} explored the allocation of public charging stations to increase the social welfare associated with transportation and power networks. Considering users' daily travel,  \cite{zhu2018charging} introduced a novel model to determine EVCS locations while minimizing the charging station installation and management costs. The authors of \cite{xi2013simulation} developed a simulation-optimization model for EVCSs to maximize the service level for EV drivers. The results show that a combination of level 1 and level 2 chargers is more desirable than installing only level 1 chargers. The research in \cite{cavadas2015mip} addressed the EVCS problem in an urban area. The authors proposed a mixed integer programming (MIP) model for locating slow-charging stations. They considered travelers' parking locations as well as their daily activities to aggregate the demand. An optimization model based on travel behavior to optimally install charging stations was developed in \cite{shahraki2015optimal}. The research in \cite{wang2013locating} used an MIP model to determine the locations for multiple types of charging stations. The results indicated that an increase in EV ranges allows installing fewer charging stations. The authors of \cite{lam2014electric} formulated a charging station location problem with a focus on human factors.
To support recent developments in the electrification of public transportation, the authors of  \cite{cai2014siting,li2017improving}, and \cite{wang2016electric} developed models to optimally determine the locations of charging stations for electric taxis and buses. {\sv The impact of different types of EVs  (\cite{fernandez2010assessment,tushar2015cost}) , locations and sizes of charging infrastructures (\cite{luo2015placement,sadeghi2014optimal,sheppard2016cost,xu2013optimal}) on power networks has also been investigated by various studies. Furthermore, various concerns from both the traffic system and power system perspectives are addressed by few studies \cite{wang2013traffic,yao2014multi,yu2015balancing}. }
\\
\indent
Even though it is important to consider uncertainties for strategic and tactical planning, as decisions made using deterministic parameters can under- or overestimate the reality \cite{birge2011introduction}, only a few research studies consider uncertainties for EV infrastructure planning. The research in \cite{faridimehr2018stochastic} developed a decision support system consisting of a modeling framework using a stochastic model and the Monte Carlo sampling method to optimally design an EV charging network. The researchers considered uncertainties in SOCs, dwell times, demand distribution, driver preferences regarding charging, the market penetration of EVs, and also drivers' willingness to walk. They used SAA and a heuristic to tackle the computational intractability of the stochastic model. The present research extends this work by considering different types of chargers and their associated preferences by the EV drivers. The authors of \cite{pan2010locating} developed a two-stage stochastic model for locating charging stations to support both the transportation system and the power grid. They considered uncertainty in the demand for batteries, loads, and generation of renewable power sources. The research in \cite{hosseini2015refueling} incorporated uncertainty regarding the traffic flow into both capacitated and uncapacitated versions of a two-stage stochastic model to locate EVCSs. With the objective of maximizing both the miles traveled by EVs and environmental benefits, the research in \cite{arslan2016benders} presented the EVCS problem as an extension of the flow refueling location problem. The authors considered both hybrid and single-fueled vehicles, and they proposed using Benders' decomposition approach for solving large-scale instances. Accounting for EV drivers' route choice behavior, \cite{riemann2015optimal} suggested a flow-capturing model with a stochastic user equilibrium to locate wireless charging infrastructures.
\begin{comment}
A multi-period optimization model was proposed by \cite{li2016multi} to capture the dynamics in the topological structure of the network. 
\end{comment}

\indent 
Charging behavior has been studied by numerous authors from different perspectives, which are multifarious amongst drivers \cite{zoepf2013charging,franke2013understanding}. To develop models that evaluate EV drivers' preferences for charging services, it is necessary to understand individuals' behaviors \cite{daina2017electric,xiong2017optimal}. The authors of \cite{xu2017joint} developed a mixed logit model to explore the factors that affect BEV users in Japan. They considered fast and normal types of chargers along with specific locations such as home, company, and public stations for installing {\sv chargers}. They identified battery capacities and initial states of charge as the main predictors for drivers' charging and location choices. The research in \cite{liu2017locating} implemented a tri-level design that considers consumers' charging and routing choices to locate multiple levels of charging facilities, including wireless charging. Findings from a national survey showed that recharging times have a considerable influence on consumers' preferences \cite{hidrue2011willingness}. The effects of policies on charging behavior and EV adoption were studied in \cite{wolbertus2018policy}. The authors used a large data set to investigate the influence of daytime and free parking policies on EV drivers' charging behaviors. The research in \cite{sun2015charge} focused on charging time behavior using a mixed logit model, and the predictors related to charging or not charging were SOCs, number of days between charging, and kilometers of travel. The results show that fast charging is preferred to normal charging. The authors of \cite{he2015deploying} suggested a tour-based BEV network equilibrium model to evaluate drivers' behaviors. Recently, the authors of \cite{hardman2018review} published a literature review on consumers' preferences for plug-in vehicle charging stations. They focused on approaches related to the expansion of {\sv charging} infrastructures based on users' interactions with EV charging stations.
\newline
\indent Choice models have recently been proposed for various purposes. The authors of \cite{benati2002maximum} proposed integrating a choice model within an optimization framework for locating new facilities in a competitive market. They used a random utility model to model customers' behavior with the aim of predicting the market shares of the locations.  In \cite{ krohn2016preventive}, the authors considered clients' utility functions, with waiting time for an appointment and the quality of care used as variables for determining health-care facility locations. Similarly, the authors of \cite{ garcia2015robust} applied a robust approach to selecting new housing programs. They incorporated a utility function with a linear combination of the features of locations and their values for potential buyers in a mathematical model with the aim of maximizing customers' satisfaction.\\
\indent Only a few studies have included multiple types of charging stations in their mathematical models (\cite{wang2013locating,cui2018locating,you2014hybrid}), and none of these have considered EV drivers' charging behavior in locating multiple types of charging stations. To the best of the authors' knowledge, the present study is the first attempt to embed a choice model within a two-stage stochastic programming approach. Also, although many studies have considered the EVCS location problem for state-wide networks (\cite{you2014hybrid,riemann2015optimal,chung2015multi,mak2013infrastructure,bayram2013decentralized,marmaras2017simulation}), only a few (\cite{faridimehr2018stochastic,zhang2015integrated}) have investigated the problem for an urban area.
\section{Preprocessing} \label{Pre}
In this study, we consider parking facilities as potential candidates for installing {\sv chargers}. Drivers select parking locations based on their preferences regarding walking distances to their final destinations. We assume that if {\sv chargers} are installed in any of the parking lots that are within a driver's preferred walking distance, the driver will be attracted to one of these, depending on the availability of that station at the time of arrival. If there are no charging stations within a driver's preferred walking distance, we do not consider that driver as contributing to the demand in our model. In the two-stage stochastic model, demand is a multi-variate random variable whose realizations are represented as scenarios. The randomness in the demand comes from many different sources, such as drivers' arrival time and their purpose in driving to the community, the duration of drivers' activities, the SOCs of EV batteries at the time of arrival, and the distances the drivers are willing to walk, based on demographics, community size, and weather conditions. The following subsection describes the uncertainties that affect the demand for public EV charging stations, based on previous work in \cite{faridimehr2018stochastic}.

\subsection{Uncertainties in Demand}
 \subsubsection{Dwell Time}Based on National Household Travel Survey (NHTS) data, we selected work, study, social, family, shopping, and meals as six different final destination categories for the EV drivers. The average dwell time reported for each category is shown in Fig. \ref{Dwell}. We used a Weibull distribution, as suggested in \cite{zhong2008studying}, to represent the duration of weekday and weekend activities.
  \begin{figure}[!htbp]
	\centering
	\includegraphics[scale = 0.35]{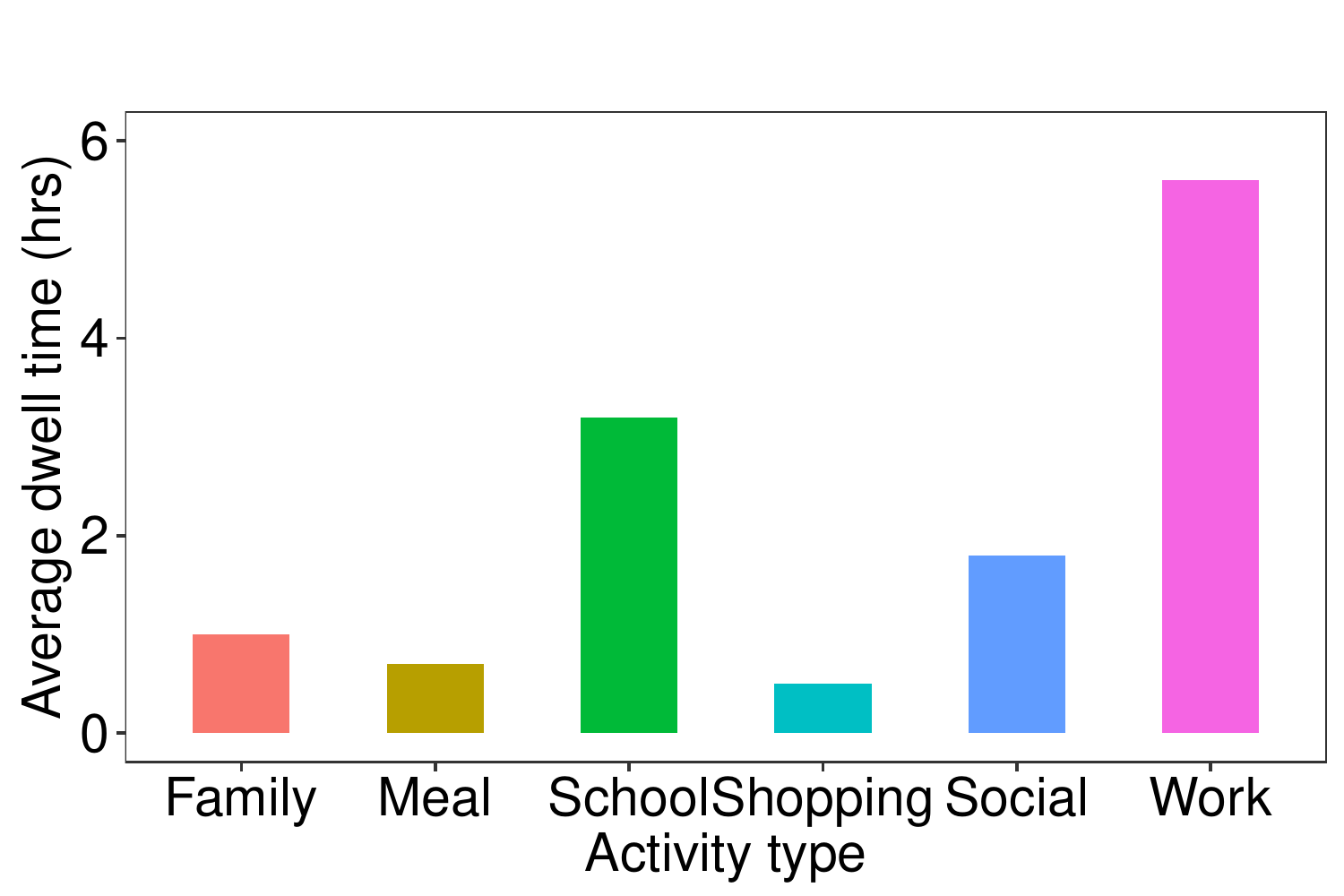}
	\captionsetup{justification=centering}
	\caption {Average dwell times for activity types; Sources: \cite{brooker2015identification} and \cite{krumm2012people}.}
	\label{Dwell}
\end{figure}
 \subsubsection{Arrival Time}EV drivers' arrival times in a community depend on the time of day, the day of the week, and the commuters' type of activity. On weekends, people tend to participate in social activities and visit shopping malls and their families more than on weekdays. On weekdays, most of the demand for {\sv chargers} comes from people who are traveling to work or school. Hence, a different demand pattern for charging stations arises on different days of the week. Fig. \ref{Arrival} shows how the demand for charging stations depends on the time and the type of day. On weekdays, the maximum demand occurs during the morning when people are arriving at work or school; in contrast, the maximum demand on weekends usually occurs around noon, when people are traveling to shopping malls and social places. The study in \cite{ozdemir2015distributed} concluded that the Weibull distribution is the best-fitting distribution for arrival times at parking lots. Therefore, we use two Weibull distributions to estimate these arrival times for weekends and weekdays.
\begin{figure}[!htbp]
	\centering
	\includegraphics[scale = 0.35]{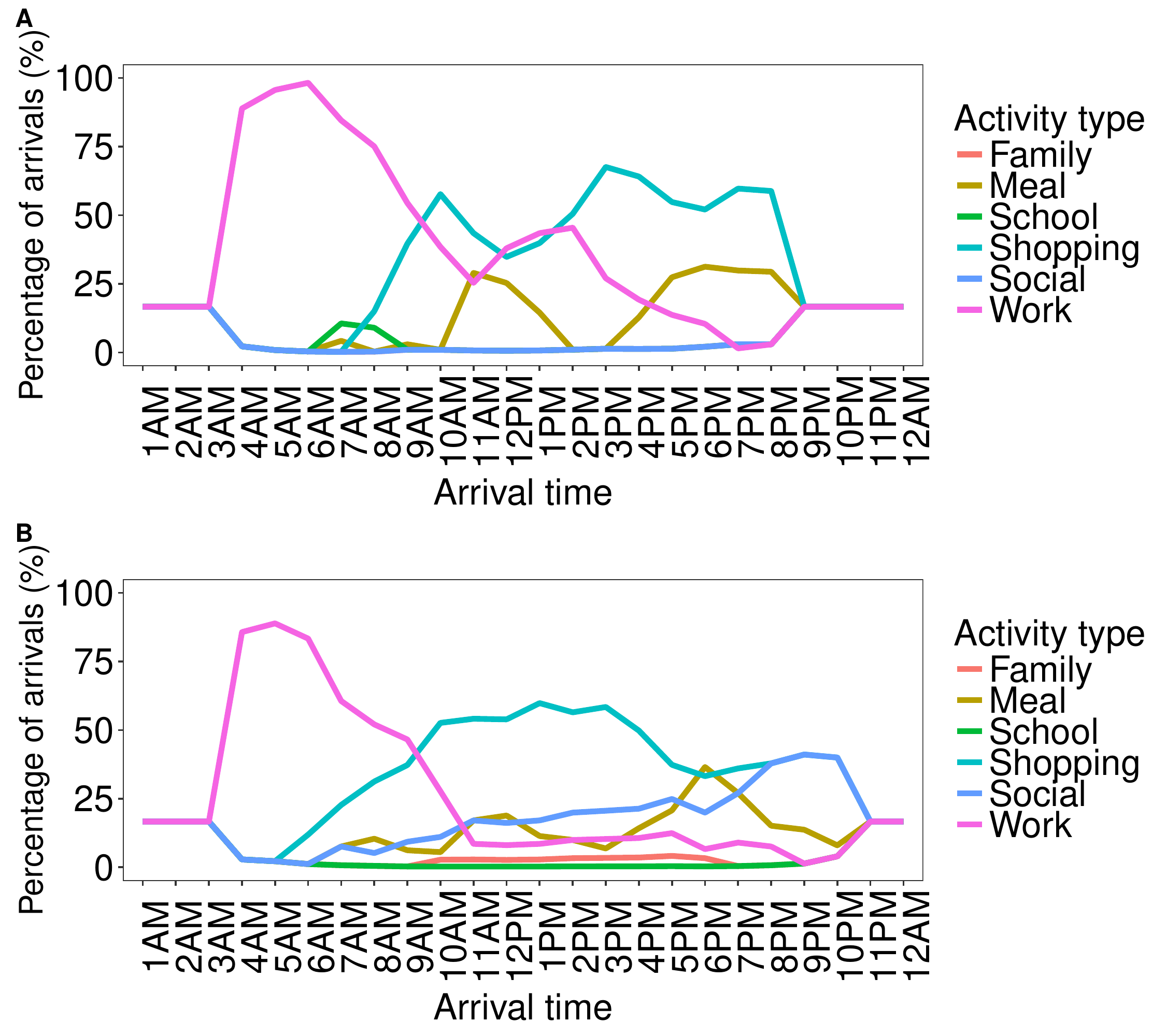}
	\captionsetup{justification=centering}
	\caption {The expected percentage breakdown for various activities by vehicle arrival times on A) weekdays and B) weekends; Sources: \cite{brooker2015identification} and \cite{krumm2012people}.}
	\label{Arrival}
\end{figure}
\subsubsection{State of charge}\label{SOC_section} While the demand for EVs is increasing due to environment- and economy-related concerns, EVs have a limited battery capacity for charging and use. Many factors, such as commuting distance, the driver's behavior, traffic congestion, and weather conditions, can affect an EV's SOC at the time of its arrival at a final destination  (\cite{you2017scheduling,torabbeigi2019drone}). Similar to \cite{fan2015operation}, we consider a normal distribution with a mean of 0.3 and a standard variance of 0.1 for the SOC of EVs when they arrive at charging locations. Fig. \ref{SOC} shows the initial SOC distribution for arriving EVs.
\begin{figure}[!htbp]
	\centering
	\includegraphics[scale = 0.35]{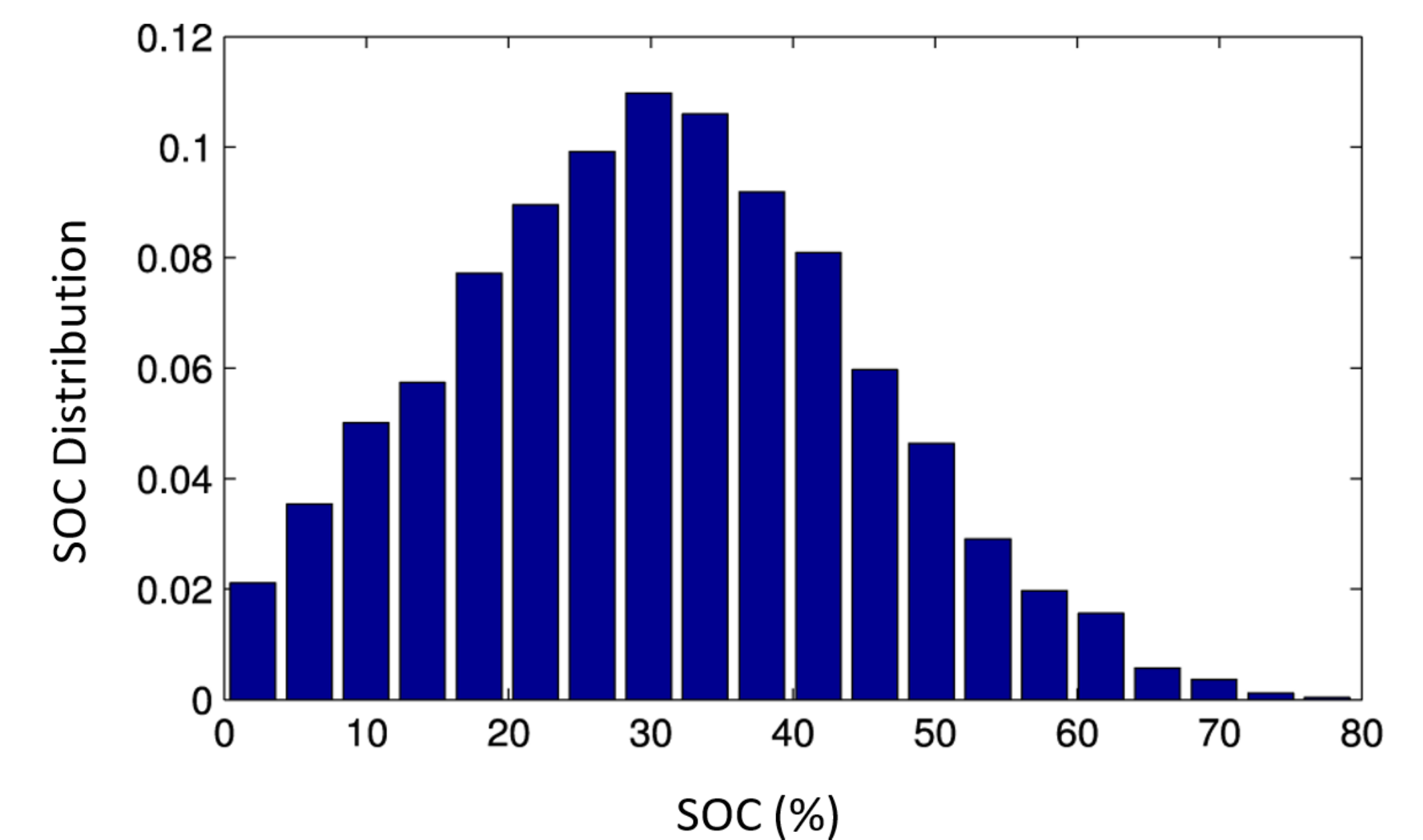}
	\captionsetup{justification=centering}
	\caption {Initial state-of-charge distribution for arriving electric vehicles; Source: \cite{fan2015operation}.}
	\label{SOC}
\end{figure}
  \subsubsection{Willingness to walk}Sociodemographic characteristics such as age, gender, education level, and occupation affect drivers' willingness to walk. Walking distances are typically shorter for children and the elderly than for the young and middle-age groups. Studies have also indicated that walking preferences are associated with many design factors, such as street connectivity, pedestrian infrastructure, and mixed land uses \cite{b99a4680343840efba8de6fc768ba7d8}. Many authors have implemented a distance decay function to illustrate individuals' willingness to walk or bicycle. The decay function parameter depends on the type of the final destination, and research using distance decay functions has also revealed different behaviors for people that live in different areas. A negative exponential distribution was used in \cite{yang2012walking} to estimate walking trips over short distances. The authors {\sv defined} the distance decay function as $P(d)= e^{-\beta \times d}$, which reflects the total percentage of walking trips for which the distance is greater than or equal to $d$ given in miles; here $\beta$ is the decay parameter. The authors used 2009 NHTS data to approximate the decay parameter $\beta$ for different groups and trip purposes. In our study, we consider the effects of the destination activity type, the season, the community size, and the region of the U.S. on drivers' walking preferences. The variation for each of these factors on the walking distance preferences estimated by \cite{yang2012walking} is shown in Fig. \ref{Walk}, and details are provided in Table \ref{WalkingPreference}.
\begin{figure}[!htbp]
	\centering
	\includegraphics[scale = 0.35]{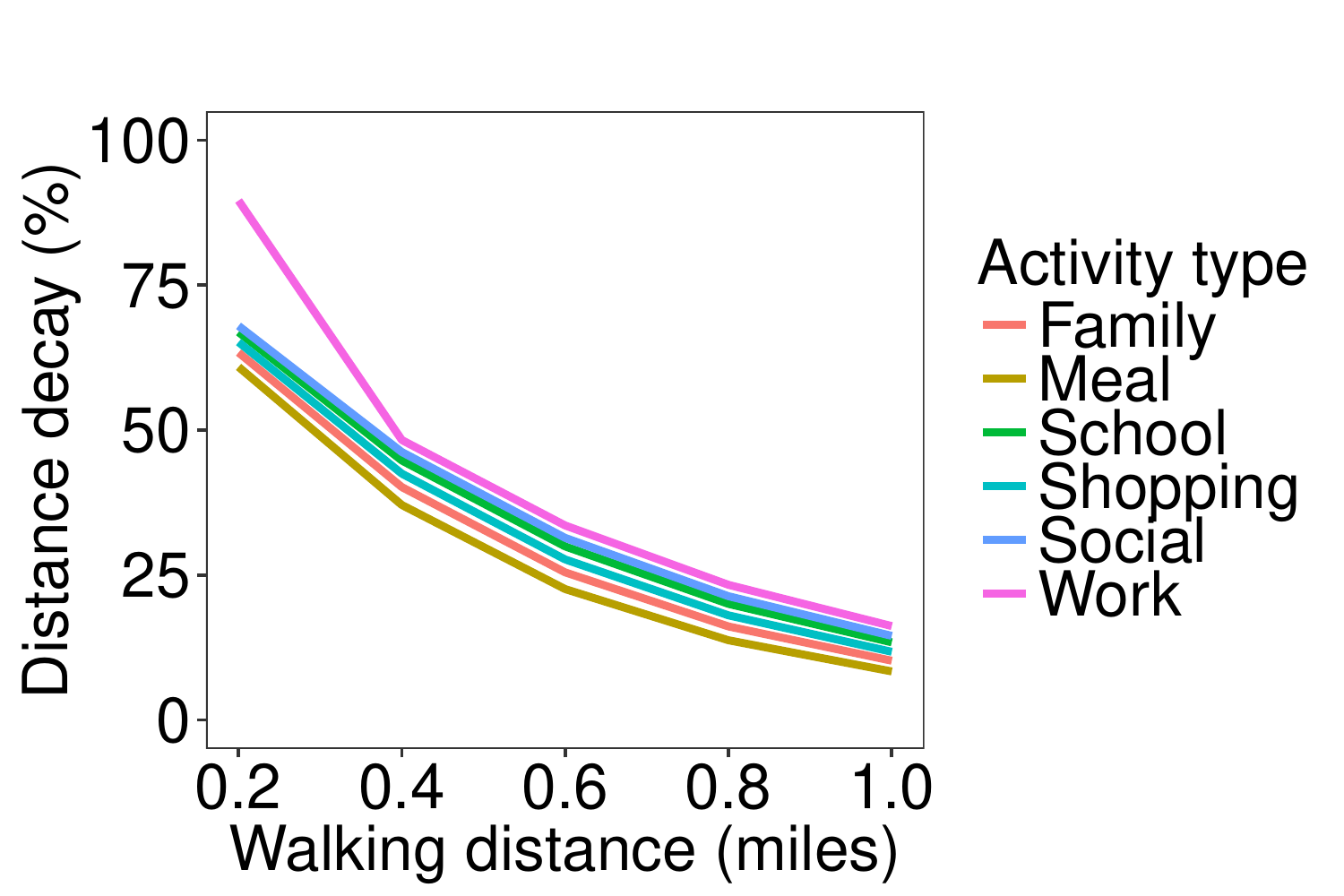}
	\captionsetup{justification=centering}
	\caption {Distance decay function for walking trips to different types of destination; Source: \cite{yang2012walking}.}
	\label{Walk}
\end{figure}
\begin{table}[!htbp]
\captionsetup{justification=centering}
	\caption{Estimated distance decay function parameters}
	\centering
	\begin{tabular}{| c | c | c |} 
		\hline
		Factor & Category & $\beta$ \\ 
		\hline
		\multirow{4}{*}{Season} & Winter & 1.88 \\
		& Spring & 1.68 \\ 
		& Summer  & 1.64 \\
		& Autumn  & 1.7 \\ 
		\hline
		\multirow{4}{*}{Region} & Northeast & 1.85 \\
		& Midwest & 1.65 \\ 
		& South & 1.76 \\
		& West & 1.65 \\ 
		\hline
		\multirow{3}{*}{Community} & Town and country & 1.68 \\
		& Suburban & 1.63 \\
		& Urban and second city & 1.78\\  
		\hline
	\end{tabular}
	\label{WalkingPreference}
\end{table}
  \subsubsection{EV market penetration}Various social, environmental and economic factors can significantly contribute to the increasing market share of different types of EVs \cite{faridimehr2018stochastic}. The research in \cite{vergis2015comparison} showed that the presence of charging infrastructure  contributes to the adoption of battery EVs but does not have any significant effect on adoption of Plug-in hybrid vehicles. The authors of \cite{liu2017uncertain} considered many sources of uncertainty in their sampling process, such as charging infrastructure availability, energy prices, and consumers' preferences. Their results project that BEV share distributions in 2030 and 2050 will have mean values of 11\% and 28\%, respectively.

\subsection{Utility Construction}
Discrete choice models are utilized to help decision makers select the best choice among different options in a choice set. These models are designed to maximize the utility of a decision maker's behaviors \cite{train2009discrete}. When an EV driver $j$ reaches a charging station, he/she can choose between $n$ different charging types that differ in terms of prices and charging {\sv duration}. A given choice among the $n$ charging types will provide an EV driver with a certain level of utility. We denote the utility that EV driver $j$ obtains from charging type $n$ as {\sv $U_{n,j}$}, $n = 1,\ldots,N$. The behavioral model will then choose charging type $n$ if and only if $U_{n,j} \geq U_{n^{'},j},\forall n,n^{'} \in N; n \neq n^{'}$. Thus, a ``utility function'' can be defined as 
$ U_{n,j} = V_{n,j} + \epsilon_{n,j}$, where  $V_{n,j} = V(X_{n,j})$ captures the deterministic part of the utility and $\epsilon_{n,j} $ is the random part capturing the non-observable variables. Some of the variables in $V$ are unknown to us, so we need to estimate them statistically. 

We consider $U_{n,j}$ to be the utility of an EV driver who is willing to charge at station $j$ using charging type $n$. Let $K$ be the set of predictor variables. Then the utility function can represented as
$$ U_{n,j} = \sum_{k \in K}\beta^{k}X^{k}_{n,j} + \epsilon^{k}_{n,j},$$
where $\beta^{k}$ are the coefficients of the corresponding variables representing the decision maker's taste.
The research in \cite{wen2016modeling} analyzed drivers' charging choices through a web-based preference survey, using a mixed logit model with various predictor variables. Table \ref{Beta} shows the estimated fixed and random effects of variables provided in \cite{wen2016modeling}.\\
\begin{table}\centering
\setlength\tabcolsep{3pt}
\captionsetup{justification=centering}
\caption{Estimated parameters using mixed logit model; Source:\cite{wen2016modeling}}
\begin {tabular}{ c  c  c }
\hline
\multicolumn{1}{c}{} & \multicolumn{1}{c}{Fixed Effects} & \multicolumn{1}{c}{Random Effects} \\  \cline{2-3}
\multicolumn{1}{c}{Variable} & \multicolumn{1}{c}{Estimate} & \multicolumn{1}{c}{Standard Deviation} \\
\hline
    Intercept & 4.756 & 0.022 \\ 
    Price  & - 0.607 & 0.089 \\
    Charging cost  & -0.062 & 0.004 \\ 
    Cost at home & 0.009 & 0.489 \\ 
    Dwell time $\geq$ 30 min & 0.335 & 0.188 \\ 
    {\sv Chargers} power (Reference: Level 1) & &  \\
    Level 2  & 1.229 & 0.253 \\ 
    Level 3 & 1.609 & 0.264 \\ 
    Ranged charged  & 0.014 & 0.003 \\
    Remaining range  & -0.130 & 0.006 \\ 
    Enough to Next Charging Opportunity  & -4.401 & 0.078 \\ 
\hline
\label{Beta}
\end{tabular}
\end{table}
Using the estimated parameters, we calculate EV drivers' utility from charging at each type of {\sv charger} and also the utility of not charging at that station. As mentioned earlier, when an EV driver arrives in the community to reach his/her final destination, a set of available parking lots is selected based on the driver's walking preferences. Then the driver's utility for each parking lot in the selected set is calculated. This process is repeated for each driver. Finally, we aggregate the utilities of the individuals to obtain the aggregated utility for each charger type in each parking lot.
\section{Notation and Model Formulation} \label{formula}
One of the common approaches to modeling a problem under uncertainty is two-stage stochastic programming. We formulate the EVCS network design problem as a scenario-based, two-stage non-linear stochastic programming model that considers the randomness arising from dwell times, drivers' willingness to walk, the EV market penetration, the demand patterns on weekdays and weekends, and SOCs. The first-stage decision variables represent ``here-and-now'' decisions that are determined based on deterministic parameters in the first-stage constraints before the uncertainty is revealed. Subsequently, second-stage decisions are determined based on the first-stage decisions and the realizations of the random variable. 

We define $J$ as the set of potential parking lots for installing a set of charger types, denoted as $N$. We define $B$ as the set of buildings that are considered to be the final destinations for EV drivers. Given $b \in B$, we define $ m \in S^M(b)$ to be collection of subsets of available parking lots within the walking preference ranges of drivers whose final destination is building $b$. We consider a collection of subsets since the EV drivers have commonality among the parking lots in reaching their final destinations. This is due to that the drivers have different walking distance preferences and hence have different parking lot subsets. We define $T$ to be the set of time slots within a day, indexed by $t \in T$. We use $\Gamma$ to denote a set of arrival and departure times, where $\gamma(a)$ and $\gamma(d)$ indicate a combination of arrival and departure times for $\gamma \in \Gamma$. We define $\tilde{\omega}$  to be a multi-variate random variable representing the demand, where each scenario $\omega$ is a realization of $\tilde{\omega}$. In the first-stage of the model, the locations and types of chargers are determined by binary variables, and the numbers of charger types in the selected parking lots are represented by integer variables. In the second-stage, based on EV drivers' walking preference ranges and the aggregated utilities for each parking lot and charger type, EV drivers are allocated to parking lots in a way that maximizes their expected access. For the mathematical formulation, we first define the model sets, parameters, and variables:
\begin{itemize}
	\item Sets
	\begin{itemize}
		\item $J$: Set of parking lots, with $j\in J$.
		\item $T$: Set of time slots, with $t \in T$.
		\item $N$: Set of charger types, with $n \in N$.
		\item $B$: Set of buildings, with $b\in B$.
		\item $S^M(b)$: Collections of subsets of possible parking lots based on the walking preferences of drivers who are going to building $b$. There are $M$ subsets and $M$ depends upon building $b$, with $m \in S^M(b)$.
		\item $\Gamma$: Set of arrival and departure times, with $\gamma \in \Gamma$.
		\item $\Omega$: Set of scenarios, with $\omega \in \Omega$.
	\end{itemize} 
	\item Model parameters
	\begin{itemize}
		\item $c_n$: Cost of installing {\sv charger} of type $n$.
		\item $k_{j}$: Capacity of parking lot $j$ for installing {\sv chargers}.
		\item $F$: Total amount of the budget for installing {\sv chargers}.
		\item $d_{\gamma,b}(\omega)$: Total demand for building $b$ between the arrival and departure times $\gamma \in \Gamma$ for a given $t \in T$ in scenario $\omega \in \Omega$.
		\item $u_{n,j}(\omega)$: The aggregated utility of EV drivers who are willing to use {\sv charger} type $n$ in parking lot $j$ in scenario $\omega \in \Omega$.	
		\item $u_{nc,j}(\omega)$: The aggregated utility of EV drivers who are not willing to charge their EVs in parking lot $j$ in scenario $\omega \in \Omega$ .
		\item $d^{'}_{\gamma,b,m}(\omega)$: The demand for building $b$ among drivers who are willing to use parking lots $m\in S^M(b)$ between the arrival and departure times $\gamma \in \Gamma$ in scenario $\omega \in \Omega $.
	\end{itemize}
	\item First-stage decision variables
	\begin{itemize}
		\item $x_{n,j}$: 1 if parking lot $j$ is chosen for installing {\sv charger} type $n$; 0 otherwise.
		\item $z_{n,j}$: Number of {\sv charger} of type $n$ in parking lot $j$.
	\end{itemize}
	\item Second-stage decision variables
	\begin{itemize}
		\item $y_{\gamma,b,j,n}^{m}(\omega)$: The proportion of the demand for building $b$ in the subsets of parking lots $S^M(b)$ between the arrival and departure times $\gamma \in \Gamma$ for a given $t \in T$ that is satisfied by parking lot $j \in S^m(b) $, where $m \in S^M(b)$, using {\sv charger} of type $n$ in scenario $\omega \in \Omega$.
	\end{itemize}
\end{itemize}
\subsection{ Two-stage Non-linear Stochastic Model}\label{model}
The two-stage non-linear stochastic programming model is defined as follows:
\begin{alignat}{3}
& \text{First-Stage Model:}  && \nonumber  \\
& \text{ Max } E_{\Omega}[\varphi(x,z,\tilde{\omega})]  && \label{fs-obj} \\ 
& \text{s.t. } && \nonumber \\
& \sum_{n \in N}{z_{n,j}}\leq k_{j} \hspace{1.6cm} \forall   j \in J, && \label{fs-obj-eq2} \\  
& z_{n,j} \leq k_{j}x_{n,j} \hspace{1.65cm} \forall   n \in N,j \in J,  &&   \label{fs-obj-eq3} \\ 
& \sum_{n \in N}\sum_{j \in J}{c_{n}z_{n,j}} \leq F  &&  \label{fs-obj-eq4} \\ 
& x_{n,j} \in \{0,1\},z_{n,j} \in \mathbb{Z^{+}} \hspace{1.5cm} \forall n\in N, j\in J. &&\label{fs-obj-eq5}
\end{alignat}
The second-stage recourse function based on the first-stage decisions $x$ and $z$ and a scenario $\omega$ is given by the following non-linear programming model:
\begin{alignat}{3}
\varphi(x,z,\omega) = &\text{ Max } 
\sum_{\gamma \in \Gamma } \sum_{b \in B} {\sum_{m \in S^M(b)}} \sum_{j \in S^m(b)}\sum_{n \in N}&& \nonumber\\
& {d_{\gamma,b}}(\omega) y_{\gamma,b,j,n}^{m}(\omega) && \label{Second-Stage Model}
\end{alignat}
\begin{alignat}{3}
& \text{s.t. } && \nonumber \\
& \sum_{\substack{\gamma \in \Gamma : \\ \gamma(a) \leq t \leq \gamma(d)} }\sum_{b \in B}\sum_{\substack{m \in S^M(b): \\ j \in S^m(b)}} {d_{\gamma,b}}(\omega)y_{\gamma,b,j,n}^{m}(\omega)\leq z_{n,j} && \nonumber \\
& \hspace{4cm} \forall t \in T,j \in J , n \in N, && \label{fs-obj-eq7}\\ 
& \sum_{\substack{m \in S^M(b): \\j\in S^m(b)}}y_{\gamma,b,j,n}^{m}(\omega)\leq  \frac{e^{u_{n,j}(\omega)}x_{n,j}}{e^{u_{nc,j}(\omega)}+\sum_{l \in N}e^{u_{l,j}(\omega)}x_{l,j}} && \nonumber \\
& \hspace{3cm} \forall \gamma \in \Gamma,b \in B,j \in J,n \in N, && \label{fs-obj-eq8}\\ 
& \sum_{n \in N}\sum_{m \in S^M(b)}\sum_{j \in S^m(b)}y_{\gamma,b,j,n}^{m}(\omega)\leq 1 \hspace{0.95cm} \forall \gamma \in \Gamma,b\in B, && \label{fs-obj-eq9} \\
&{d_{\gamma,b}}(\omega)\sum_{n \in N}\sum_{j \in S^m(b)}y_{\gamma,b,j,n}^{m}(\omega)\leq d^{'}_{\gamma,b,m}(\omega) && \nonumber \\ & \hspace{4cm} \forall \gamma \in \Gamma, b \in B, m \in S^M(b), && \label{fs-obj-eq10}\\ 
&  0 \leq y_{\gamma,b,j,n}^{m}(\omega)\leq 1 && \nonumber \\&\hspace{1cm} \forall\gamma \in \Gamma, b \in B, m \in S^M(b),  j \in S^m(b), n \in N. && \label{fs-obj-eq11}
\end{alignat}
The first-stage objective function \eqref{fs-obj} maximizes the expected EV drivers' access to the charging stations. Constraints (\ref{fs-obj-eq2}) represent capacity restrictions for each type of charger in a parking lot based on its capacity, and constraints (\ref{fs-obj-eq3}) state that a parking lot must be selected before selecting the charger type. Constraints (\ref{fs-obj-eq4}) give the budgetary constraints. Constraints (\ref{fs-obj-eq5}) define the binary and integer restrictions for the first-stage variables. For a realization of $\omega \in \Omega$, the second-stage objective function (\ref{Second-Stage Model}) maximizes the EV traffic flows based on the network decisions made in the first-stage. For each time slot in the planning horizon $t \in T$, the constraints (\ref{fs-obj-eq7}) limit access based on the capacity decided upon in the first-stage. Constraints (\ref{fs-obj-eq8}) limit EV drivers' choice of different levels of chargers based on the utility function estimated by the mixed logit model described in the previous section. Constraints (\ref{fs-obj-eq9}) ensure that the allocation of flow to the charging stations for each building does not exceed the building's demand in any time slot. Constraints (\ref{fs-obj-eq10}) guarantee that drivers are assigned to only one of the parking lots within their walking distance range. Finally, constraints (\ref{fs-obj-eq11}) define the restrictions for the second-stage variables. Due to the constraints \eqref{fs-obj-eq8}, the two-stage model is non-linear in nature and is in general difficult to solve. In the next section, we provide details for linearizing the model so that it is viable for computational efficiency. \\
\begin{proposition} \label{prop:yint-relax}
First, we restate constraints (\ref{fs-obj-eq8}) as:\\
$$\sum_{\substack{m \in S^M(b): \\j\in S^m(b)}}y_{\gamma,b,j,n}^{m}{({e^{u_{nc,a}}}+{\sum_{l \in N}}{e^{u_{l,j}}x_{l,j})\leq {e^{u_{n,a}}x_{n,j}}}}.$$
Then for bounded continuous and binary variables $y$ and $x$, respectively, we define a non-negative bi-linear variable as follows: 
\begin{alignat}{3}
& o_{\gamma,b,j,n , l}^{m} = x_{l,j} y_{\gamma,b,j,n}^{m} && \nonumber \\&\hspace{0.5cm} \forall \gamma \in \Gamma, n \in N, l \in N,b \in B,m \in S^M(b), j \in S^m(b) && \nonumber.
\end{alignat}
Using the variables $o$, a standard approach that has been adopted for linearizing the bi-linear terms is to replace each term by its convex and concave envelopes, also called the ``McCormick envelopes'' \cite{mccormick1976computability}. The constraints (\ref{fs-obj-eq8}) can then be rewritten as:\\
\begin{alignat}{3}
&e^{u_{nc,j}}\sum_{\substack{m \in S^M(b): \\j\in S^m(b)}}y_{\gamma,b,j,n}^{m} + \sum_{\substack{m \in S^M(b): \\j\in S^m(b)}}\sum_{l \in N}{e^{u_{l,j}}o_{\gamma,b,j,n,l}^{m}}  &&\nonumber\\
& \leq e^{u_{n,j}}x_{n,j} \hspace{1cm} \forall \gamma \in \Gamma,b \in B,j \in J,n \in N,
&& \label{fs-obj-eq99}
\end{alignat}
\begin{alignat}{3}
& o_{\gamma,b,j,n,l}^{m} \leq x_{n,j} &&\nonumber\\& \hspace{0.25cm} \forall \gamma \in \Gamma, n \in N, l \in N,b \in B,m \in S^M(b), j \in S^m(b),&& \label{fs-obj-eq100} \\
& o_{\gamma,b,j,n , l}^{m} \leq y_{\gamma,b,j,n}^{m}&&\nonumber\\&\hspace{0.25cm} \forall \gamma \in \Gamma, n \in N, l \in N,b \in B, m \in S^M(b), j \in S^m(b), && \label{fs-obj-eq101}\\
& o_{\gamma,b,j,n,l}^{m} \geq x_{n,j} + y_{\gamma,b,j,n}^{m} -1 &&\nonumber\\&\hspace{0.25cm} \forall \gamma \in \Gamma, n \in N, l \in N, b \in B, m \in S^M(b), j \in S^m(b). &&\label{fs-obj-eq102}
\end{alignat}
\end{proposition}
\begin{proof}
For the proof, see \cite{mccormick1976computability}.
\end{proof}

\begin{comment}

\subsection{Proposition}
Let us consider constraints (\ref{fs-obj-eq8}):\\
$$\sum_{\substack{m \in S^M(b): \\j\in S^m(b)}}y_{\gamma,b,j,n}^{m}\leq  \frac{e^{u_{n,j}}x_{n,j}}{e^{u_{nc,j}}+\sum_{l \in N}e^{u_{l,j}}x_{l,j}}, $$
and as the denominator is positive this is equivalent to :\\
\\
$$\sum_{\substack{m \in S^M(b): \\j\in S^m(b)}}y_{\gamma,b,j,n}^{m}{({e^{u_{nc,a}}}+{\sum_{l \in N}}{e^{u_{l,j}}x_{l,j})\leq {e^{u_{n,a}}x_{n,j}}}} .$$
\end{comment}
This reformulation helps represent the second-stage problem as a linear programming model, thus allowing us to use the L-shaped method as a decomposition algorithm. {\sv It is worth to mention that since the two-stage model is an extension of capacitated facility location problem, it is a NP-hard problem \cite{melkote2001capacitated}.}
\section{Methodology and Algorithm Development}\label{Method}
\subsection{Sample Average Approximation}
The SAA method is an approach to solving two-stage stochastic programming problems that uses Monte Carlo simulation. It is a sampling technique for approximating the expectation function in a two-stage model. SAA approximates the second-stage expected recourse function of the two-stage stochastic programming model by a sample average estimate derived from a random sample. Then the sample average approximating the two-stage model is solved using a decomposition algorithm or a direct solver. The SAA model is solved multiple times with different samples to obtain candidate solutions along with statistical estimates of their optimality gaps. The SAA procedure is specified in Algorithm \ref{SAA_procedure}. 
 \begin{algorithm}[!htbp]
	\caption{: SAA}
	\begin{algorithmic}
    	\State \textbf{Estimate the upper bound:}
    	\State \hskip1.5em Generate $K$ independent sample sets of scenarios, each
    	\State \hskip1.8em  of size $L$, i.e., \:($\omega^1_j,\omega^2_j,...,\omega^L_j$) for $j = 1,2,...,K$.
    	\State \hskip1.5em For each sample set $j = 1,2,...,K$, find the optimal
    	\State \hskip1.8em solution:
    	\State \hskip1.5em $$v^j_{L} = \frac{1}{L} \sum_{i=1}^{L} \varphi(x,z,\omega^i_j).$$
    	\State \hskip1.5em Calculate: $$\overline v_{L,K} = \frac{1}{K} \sum_{j=1}^{K} v^j_{L},$$
    	\State \hskip1.5em $$\sigma^2_{\overline v_{L,K}} = \frac{1}{K(K-1)} \sum_{j=1}^{K} (v^j_{L} - \overline v_{L,K})^2.$$
    	\State \textbf{Estimate the lower bound:}
    	\State \hskip1.5em Choose any feasible solution ($\overline x,\overline z$) from the first-stage
    	\State \hskip1.8em problem, which provides a lower bound for the optimal
    	\State \hskip1.8em value $f(\overline x,\overline z) \leq v^*$. 
    	\State \hskip1.5em Choose a sample of scenarios of a size $L'$ that is much
    	\State \hskip1.8em larger than $L$ and independent of the samples, i.e.,
    	\State \hskip1.8em ($\omega^1,\omega^2,...,\omega^{L'}$).
    	\State \hskip1.5em Estimate the objective function $f$:
    	\State \hskip1.5em
    	\State \hskip5em $f(\overline x, \overline z) = \frac{1}{L'} \sum_{i=1}^{L'} \varphi(x,z,\omega^i).$
    	\State \hskip1.5em
    	\State \hskip1.5em Calculate the variance of this estimation:
    	\State 
    	\State \hskip1.5em $\sigma^2_{L'}(\overline x,\overline z) = \frac{1}{L'(L'-1)} \sum_{i=1}^{L'} (\varphi(x,z,\omega^i)
		- f(\overline x, \overline z))^2.$
		\State 
		\State \textbf{Estimate the optimality gap and variance:}
		\State Based on the computed upper and lower bounds, the optimality gap is estimated as follows:
		\State $$Gap_{K,L,L'}(\overline x,\overline z) = \overline v_{L,K} - f(\overline x,\overline z).$$
		\State Similarly, the variance is calculated as follows:
		\State $$\sigma^2_{gap} = \sigma^2_{\overline v_{L,K}} + \sigma^2_{L'}(\overline x,\overline z).$$
    \end{algorithmic}
    \label{SAA_procedure}
\end{algorithm}

The SAA procedure for statistical evaluation of a candidate solution was suggested in \cite{mak1999monte}, while convergence properties for the SAA method were studied in \cite{kleywegt2002sample}.

\begin{table}[!htbp]
    \centering
	\caption{SAA performance}
	\begin{tabular}{|c|c|c c  c c|} 
		\hline
		$S$ & $P$ & UB & LB & Gap & SD \\
		\hline
		\multirow{4}{*}{10} & 5 & 230.78 & 224.40 & 6.38 & 5.57\\
		&10 & 246.89 & 242.83 & 4.06 & 4.98 \\
		& 15 & 277.43 & 272.72 & 4.71 & 3.55 \\
		& 20 & 300.76 & 295.42 & 5.34 & 3.76\\
		\hline
		\multirow{4}{*}{20} & 5 & 227.90
		& 224.90 & 3.00 & 6.11 \\
		& 10 & 268.78 & 264.95 & 3.83 & 5.20 \\
		& 15 & 296.59 & 291.80 & 4.79 & 3.65 \\
		& 20 & 310.49 & 306.78 & 3.71 & 4.76 \\
		\hline
		\multirow{4}{*}{30} & 5 & 229.75
		& 226.23 & 3.52 & 2.03\\
		&10& 265.20 & 261.82 & 3.38 & 2.98 \\
		& 15 & 286.43 & 285.11 & 1.32& 3.81\\
		& 20 & 273.18 & 272.28 & 0.90 & 2.48\\
		\hline
		\multirow{4}{*}{40} & 5 & 227.29
		& 226.40 & 0.89 & 2.75\\
		&10 & 278.13 & 277.00 & 1.13 & 2.21 \\
		& 15 & 304.46 & 302.31 & 2.15 & 1.17\\
		& 20 & 323.39 & 322.85 & 0.54 & 1.88\\
		\hline
		\multirow{4}{*}{50} & 5 & 220.10& 219.70 & 0.40 & 2.90\\
		&10 & 289.42 & 288.21 & 1.21 & 3.12 \\
		& 15 & 308.24 & 307.90 & 0.34 & 2.26\\
		& 20 & 322.00 & 321.85 & 0.15 & 2.59\\
		\hline
	\end{tabular}
	\label{SAA}
\end{table}

Table \ref{SAA} presents the computational results for the two-stage model using the SAA procedure.  In the table, `S' and `P' represent the numbers of scenarios and parking lots, respectively. The upper and lower bounds are represented as `LB' and `UB', respectively. The upper bound for the expected accessibility of the charging station is estimated by a batch size of 20 ($K$=20). An independent sample of scenarios ($L^{'}$=1,000) were used to estimate a lower bound for the optimal solution. The gap and standard deviation are represented in the columns `Gap' and `SD', respectively.
\subsection{L-shaped Decomposition}
SAA was adopted for the model presented in section \ref{model}, and two-stage sample average stochastic programs are commonly solved by decomposition algorithms such as Benders' method and the L-shaped method. Realistic problems are continuously growing in size and complexity; for this reason, decomposition techniques are more attractive. Decomposition methods break a problem down into smaller problems that are easier to solve. The L-shaped method has been applied to the class of mixed-integer linear stochastic programming problems with only continuous variables in the second-stage. The L-shaped method works by approximating the expected second-stage recourse function through construction of optimality cuts in the first-stage based on the dual solutions of the second-stage problems. The procedure alternates between a master problem (MP), as represented in \eqref{MP-FSl}, and sub-problems (SPs), trading information to obtain the optimal solution. The SPs are the second-stage formulation \eqref{Second-Stage Model}, subject to constraints \eqref{fs-obj-eq7}-\eqref{fs-obj-eq11}. 

\begin{alignat}{3}
& \text{Master Problem (MP):}  && \nonumber  \label{MP-FSl}\\
& \text{ Max } \eta  && \\ 
& \text{s.t. } && \nonumber \\
& \eqref{fs-obj-eq2} - \eqref{fs-obj-eq5}, \, \, \,\eta \,\,\,\, \text{free}.\nonumber  
\end{alignat}

In the problem, constraints (\ref{fs-obj-eq7}), (\ref{fs-obj-eq99}), (\ref{fs-obj-eq100}), and (\ref{fs-obj-eq102}) are referred to as the linking constraints because of the presence of the first-stage variables $x$ and $z$ in the second-stage, which links the two stages. Let \textbf{$A$} be the coefficient matrix for variables $z_{n,j}$ in the linking constraints (\ref{fs-obj-eq7}), where $a_{i,q}$ is the entry of matrix \textbf{$A$} at indices $i$ and $q$, and $a_{i,q} \in R^{|N||J||T|\times|N||J|}$. In addition, let $F$ and $G$ be the coefficient matrices for variables $x_{n,j}$ in the linking constraints (\ref{fs-obj-eq99}) and ((\ref{fs-obj-eq100}), (\ref{fs-obj-eq102})), respectively, where $f_{i,q}$ and $g_{i,q}$ are the entries of the matrices \textbf{$F$} and \textbf{$G$} at indices $i$ and $q$,  and $f_{i,q} \in R^{|N||J||T||B|\times|N||J|}$, $g_{i,q} \in R^{2|T||N||N||B||J||S|\times|N||J|}$. We use $\pi_{\omega}$ as the notation for a vector for the dual values of the second-stage constraints ((\ref{fs-obj-eq7}), (\ref{fs-obj-eq9}), (\ref{fs-obj-eq10}), (\ref{fs-obj-eq11}) (\ref{fs-obj-eq99}), (\ref{fs-obj-eq100}), (\ref{fs-obj-eq101}), (\ref{fs-obj-eq102})), and $\pi_{\omega}^{1}$, $\pi_{\omega}^2$, and $\pi_{\omega}^{3}$ as dual values corresponding to the constraints (\ref{fs-obj-eq7}), (\ref{fs-obj-eq99}) and ((\ref{fs-obj-eq100}), (\ref{fs-obj-eq102})) in each scenario. We use $\Delta_{\omega}$ to refer to the right-hand sides of the SPs in each scenario. The L-shaped method is initialized by solving the first-stage EVCS problem to obtain the initial solutions $x^{0}$ and $z^{0}$. These solutions are then used as fixed parameters in the second-stage problem. For each scenario $\omega \in \Omega$, a sub-problem is defined based on the second-stage problem. In the next step, the SPs are solved to obtain the dual values and the corresponding objective functions. The optimality cut(s) is (are)  then generated using matrix multiplication. For each scenario, an optimality cut can be defined as follows:
\begin{alignat}{3}
&\sum_{n \in N}\sum_{j \in J}\Big(((\pi^1)^{T}A)\cdot z_{n,j} + \big( ((\pi^2)^{T}F) ((\pi^3)^{T}G)\big)\cdot x_{n,j} \Big)  &&\nonumber \\ 
& +\eta \leq \pi^{T}\cdot \Delta &&\nonumber,\label{lsh}
\end{alignat}
 where $\eta$ is a free variable. We use $\Theta_{\omega}^{k}$ to denote the optimality cut corresponding to scenario $\omega$ at iteration $k$. It should be noted that because the second-stage is feasible for every solution of the first-stage (complete recourse), we do not need to add any feasibility cuts in \eqref{MP-FSl}. In the next step, the generated cut(s) are added to the MP \eqref{MP-FSl} with the objective function to maximize $\eta$ for the single-cut and $\sum_{{\omega \in \Omega}} p_\omega \eta_{\omega}$ for the multi-cut L-shaped decomposition where $p_\omega$ is the probability of occurrence for each scenario $\omega$. Then the MP is solved to obtain a new solution for the variables $x$ and $z$, and the updated solution is then added to the sub-problems. In each iteration, upper and lower bounds are updated based on the new solutions obtained from the sub-problems and the MP. This process is repeated until the difference between the upper and lower bounds reaches a pre-determined threshold. 
\newline
\indent To evaluate the efficacy of the L-shaped algorithm, we conducted computational experiments with various instances. We implemented single- and multi-cut L-shaped decomposition methods to solve the large-scale sample average two-stage stochastic programming models. We compared the performance of these two methods to the deterministic equivalent problem (DEP). The DEP is the entire representation of formulation \eqref{fs-obj}-\eqref{fs-obj-eq11} without any decomposition for the problem. Table \ref{Decomposition_1} indicates the complexity of instances in terms of the number of variables and constraints in the first-stage and the second-stage, along with the number of non-zeros. The columns labelled `$S$', `$P$', `Cons', and `Vars' represent the number of scenarios, parking lots, constraints, and variables, respectively.
\begin{algorithm}[!htbp]
	\caption{: L-shaped decomposition}
	\begin{algorithmic}
    	\State \textbf{Initialization:}
    	\State \hskip1.5em Obtain an initial solution $z^{0}$ and $x^{0}$  by solving the first-
    	\State \hskip1.8em stage problem.
    	\State \hskip1.5em Set UB $\gets$ $\infty$, LB $\gets$ $-\infty$, $k \gets 0$.
    	\State \hskip1.5em Define the free variable $\eta$ for single-cut  and  $\eta_\omega$ for \State \hskip1.8em multi-cut decomposition.
    	\State \textbf{While} UB - LB $> \epsilon$:
    	\State \hskip2em\textbf{Sub-problems:}
    	\State \hskip3em For $\forall \omega \in \Omega$:
    	\State \hskip5em Solve $\varphi(x,z,\omega)$.
    	\State \hskip5em Calculate the dual solution for $\varphi(x,z,\omega)$ 
    	\State \hskip5.3em and store it as $\pi^{k}$.
    	\State \hskip5em Extract $\pi^{1,k}$, $\pi^{2,k}$, and $\pi^{3,k}$ from $\pi^{k}$. 
    	\State \hskip5em Calculate the objective function value for 
    	\State \hskip5.3em $\varphi(x,z,\omega)$ and store it as $f^{k}_{\omega}$.
    	\State \hskip2em \textbf{Update upper bound:}
    	\State \hskip3.5em Set $v^{k} = \sum_{\omega \in \Omega}p_{\omega}f^{k}_{\omega}$, where $p_{\omega}$ is the probability
    	\State \hskip3.8em of the occurrence of scenario $\omega \in \Omega$.
        \State \hskip3.5em Set  UB $\gets$ min (UB, $v^{k}$).
    	\State \hskip2em \textbf{Cut generation:}
    	\State \hskip3.5em \textbf{Single-cut:}
    	\State \hskip5em Add $\sum_{\omega \in \Omega} p_{\omega} \Theta^k_{\omega}$ to the first-stage problem.
    	\State \hskip3.5em \textbf{Multi-cut:}
    	\State \hskip 5em $\forall \omega \in \Omega:$
    	\State \hskip 7em Add $\Theta^k_{\omega}$ to the first-stage problem.
    	\State \hskip2em \textbf{Master problem:}
		\State \hskip3em Set $v^{k+1} \gets \eta$ as the objective function for
		\State \hskip3.3em single-cut.
		\State \hskip3em Set  $v^{k+1} \gets \sum_{{\omega \in \Omega}}\eta_{\omega}$ as the objective function for
		\State \hskip3.3em multi-cut.
    	\State \hskip3em  Solve the MP and update $z^{*} \gets z^k$ and $x^{*} \gets x^k$.
    	\State \hskip2em \textbf{Update lower bound:}
    	\State \hskip3.5em Set  LB $\gets$ max (LB,$v^{k+1}$). 
    	\State \hskip2em Set $k\gets k+1$.
    \end{algorithmic}
\end{algorithm}
 \begin{table*}[!htbp] 
	\caption{Model data specifications}
	\centering
	{\begin{tabular}{|c|c| c c c c c c c|} 
		\hline
		$S$ & $P$ & Cons & Vars & First-Stage Vars &First-Stage Cons & Second-Stage Vars &Second-Stage Cons & \# of Non-zeros \\
		\hline
		\multirow{4}{*}{10} & 5 & 477,970
 & 178,770
 & 15 & 21 & 447,955 & 178,749 & 2,886,750  \\
		& 10 & 951,240 & 357,540 & 30 & 41 & 951,210 & 357,499& 7,987,222   \\
		& 15& 1,424,510 & 536,310 & 45 & 61 & 1,424,465 & 536,249& 12,342,786 \\
		& 20 & 1,897,780 & 1,715,080 & 60 & 81 & 1,897,720 & 714,999& 18,967,552  \\
		\hline
		\multirow{4}{*}{20} & 5 & 1,043,334
 & 390,210 & 15 & 21 & 1,043,319 & 390,189 &  4,245,768 \\
		& 10 & 2,076,404
 & 780,360 & 30 & 41 & 2,076,374 & 780,319& 11,879,054 \\
		& 15& 3,109,474 & 1,170,630 & 45 & 61 & 3,109,429 & 1,170,569 & 17,652,320  \\
		& 20& 4,142,544 & 1,560,840 & 60 & 81 & 4,142,484  &  1,560,759 & 25,657,932\\
		\hline
		\multirow{4}{*}{25} & 5 & 1,244,946

 & 465,090 & 15 & 21 &1,244,931 & 465,069 & 6,676,510\\
		& 10 & 2,477,576 & 930,180
 &30&41 & 2,477,546 & 930,139 & 15,777,890   \\
		& 15& 3,710,206 & 1,395,270 &45 & 61 & 3,710,161 & 1,395,209  & 21,876,112 \\
		& 20& 4,942,836 & 1,860,360 & 60 & 81 & 4,942,776 & 1,860,279 & 33,132,981 \\
		\hline
	\multirow{4}{*}{30} & 5 & 1,456,852
 & 543,450 & 15 & 21 & 1,456,837 & 543,429 & 8,352,947  \\
		& 10 & 2,899,182 & 1,086,900
 & 30 &41 &2,899,152 & 1,086,859 & 20,301,290  \\
		& 15& 4341,512 &1,630,350 & 45 & 61 & 4341,467 &1,630,289 & 37,392,389 \\
		& 20& 5,783,842 & 2,173,800 & 60 & 81  &5,783,782 & 2,173,719 & 75,390,221 \\
		\hline
		\multirow{4}{*}{35} & 5 & 1,750,951
 & 653,910 & 15 & 21 & 1,750,936 & 653,889&  10,893,269\\
		& 10 & 3,484,551 & 1,307,820
 & 30 & 41 & 3,484,521 & 1,307,779 & 21,290,765 \\
		& 15& 5,218,151 & 1,961,730
 & 45 & 61 &  5,218,106 & 1,961,669  & 45,888,242 \\
		& 20& 6,951,751 & 2,615,640 & 60 & 81 &6,951,691 & 2,615,559  &  80,561,107 \\
		\hline
			\multirow{4}{*}{40} & 5 & 2,093,635
 & 781,590 & 15 & 21 & 2,093,620 & 781,569 & 15,896,110 \\
		& 10 & 4,166,475
 & 1,563,180 & 30 & 41 &4,166,445 & 1,563,139  &  30,290,137  \\
		& 15& 6,239,315 & 2,344,770 & 45 & 61 & 6,239,270 & 2,344,709 & 59,876,208 \\
		& 20& 8,312,155
 & 3,126,360 & 60 & 81  & 8,312,095 & 3,126,279 & 101,965,108  \\
		\hline
	\end{tabular}}
	\label{Decomposition_1}
\end{table*} 
 \begin{table*}[!htbp] 
	\caption{Computational results for L-shaped method}
	\centering
	\begin{tabular}{|c|c| c c| c c c| c c c|} 
		\hline
		\multirow{2}{*}{$S$} & \multirow{2}{*}{$P$}& \multicolumn{2}{c|}{DEP} & \multicolumn{3}{c|}{Single-cut}& \multicolumn{3}{c|}{Multi-cut}
		\\\cline{3-4} \cline{5-7} \cline{8-10}
		& & {\sv time(s)} & {\sv gap(\%)} & {\sv time(s)} & {\sv gap(\%)} & \# of cuts & {\sv time(s)} & {\sv gap(\%)} & \# of cuts \\
		\hline
		\multirow{4}{*}{10} & 5 & 303
 & 0.00
 & 1,234& 0.00 & 453 & 204 & 0.00& 570 \\
		& 10 & 2,263 & 0.00 & 3,600 & 0.30 & 598 & 539 & 0.00 & 780 \\
		& 15& 2,779 & 0.00 & 3,600& 0.20 &438 & 321 & 0.00& 510 \\
		& 20& 3,600 & 0.10 & 3,600 & 3.00 &253& 3,600 & 2.00 & 1,750 \\
		\hline
		\multirow{4}{*}{20} & 5 & 1,845
 & 0.00& 2,747 & 0.00 & 349 & 432 & 0.00 & 1,060 \\
		& 10 & 3200
 & 0.00 & 3,600 & 0.70 & 366 & 1,036 & 0.00 & 1,420 \\
		& 15& 3,600 & 10.80 & 3,600 & 3.90 & 272 & 737 & 0.00 & 1,180 \\
		& 20& 3,600 & 0.20 & 3,600 & 33.00 & 162 &  3,600 & 0.70& 1,360\\
		\hline
		\multirow{4}{*}{25} & 5 & 2,575

 & 0& 3,600 & 0.04& 508 & 921 & 0.00 & 1,675 \\
		& 10 & 3,600 & - & 3,600 & 4.10 & 301 & 3,600 & 0.20&3,300 \\
		& 15& 3,600 & - & 3,600 & 5.20 & 251 & 3,600 & 0.30& 3,125 \\
		& 20& 3,600 & 0.20 & 3,600 & 12.80 & 153 & 3,600 & 0.50&2,775\\
		\hline
	\multirow{4}{*}{30} & 5 & 3,057 & 0.00 & 3,600 & 0.00 & 467 & 661 & 0.00& 1,530\\
		& 10 & 3,600 & - & 3,600 & 62.40 & 499 & 1,086 & 0.00 & 1,710\\
		& 15& 3,600 & - & 3,600 & 5.80& 201 & 1,198 & 0.00& 1,860\\
		& 20 & 3,600 & 2.40 & 3,600 & 46.25 & 116 & 3,600 & 6.00& 2,130\\
		\hline
		\multirow{4}{*}{35} & 5 & 3,200 & 0.00 & 3,600 & 0.20 & 207& 1,401 & 0.00 & 1,610\\
		& 10 & 3,600 & - & 3,600 & 11.60 & 177 & 3,600 & 0.02& 2,380 \\
		& 15& 3,600 & - & 3,600 & 5.10 & 152 & 3,600 & 0.01& 2,205\\
		& 20& 3,600 & - & 3,600 & 11.80 & 98 & 3,600 & 0.15&1,575\\
		\hline
			\multirow{4}{*}{40} & 5 & 3,600 & - & 3,600 & 0.15 & 203 & 3,600 & 0.11 & 2,240\\
		& 10 & 3,600 & - & 3,600 & 2.40 & 148 & 3,600 & 0.22& 1,680 \\
		& 15 & 3,600 & - & 3,600 & 9.00 & 137 & 1,414 & 0.00 & 840\\
		& 20& 3,600 & - & 3,600 & 37.00& 76 & 3,600 & 4.00 & 400\\
		\hline
	\end{tabular}
	
	\label{Decomposition_2}
\end{table*}
Table \ref{Decomposition_2} shows the numerical results. {\sv The first two columns indicate the performance of the DEP in terms of runtime in seconds and the MIP gap ($\%$). The next three columns specify the performance of single-cut L-shaped decomposition, with `time(s),' `gap($\%$)' and `$\#$ of cuts' indicating runtime (seconds), the gap percentage, and the total number of cuts within the stipulated time limit, respectively. Similarly the last three columns indicate the performance metrics related to multi-cut L-shaped}. The gap percentage is calculated as the difference between the upper bound and the lower bound divided by the lower bound. Similarly, the last three columns specify the performance of multi-cut L-shaped decomposition. All of the optimization models were implemented in Python 3.6 using Gurobi 8.1.1, with a one-hour time limit. The computational experiments were performed on a computer with an Intel \textregistered \, Xeon \textregistered \, CPU E5-2640, 2.60 GHz, and 80GB RAM. As shown in Table \ref{Decomposition_2}, the runtime for most instances increased with an increase in the number of parking lots and  scenarios, as expected. When the number of scenarios was less than 20, DEP performed better than the single-cut method. However, for the rest of the 16 instances, DEP could obtain a feasible solution within the one-hour time limit in only four instances, while single-cut decomposition performed better in most of these instances. Especially for the large-scale instances, multi-cut decomposition outperformed DEP with a much better runtime and gap. Also, multi-cut decomposition outperformed single-cut decomposition in all instances. Given the relatively simple and fewer constraints in the first-stage model, the multi-cut variant was able to perform better than single-cut. The single-cut L-shaped method mostly had difficulties in accelerating the convergence of the upper and lower bounds. Hence, for any given data set, multi-cut decomposition outperformed the other methods, and especially when there were a larger number of parking lots and scenarios.
 \subsection{Value of the Stochastic Solution}
The utility of the stochastic programming approach can be evaluated by estimating the value of the stochastic solution (VSS) introduced by \cite{birge1982value}. The objective value of the recourse problem (RP) can be stated as RP=$E_{\Omega}[\varphi(x,z,\tilde{\omega})]$; then we take the expected value of the random variable and solve the \textit{expected value problem}, EV=$ \varphi(x,z,\bar{\omega})$, where $\bar{\omega}$ for the demand parameter is  $\sum_{\omega \in \Omega} p_{\omega} d_{\gamma,b}(\omega)$, with $p_{\omega}$ indicating a scenario $\omega$'s probability of occurrence and $\sum_{\omega \in \Omega}p_{\omega}$=1. Considering $\bar{x},\bar{z}$ as the solutions for the EV problem, the expected result of using the expected value solutions $(\bar{x},\bar{z})$ is EEV=$E_{\Omega}[\varphi(\bar{x},\bar{z},\tilde{\omega})]$. Then the VSS can be defined as the difference between the objective values of the recourse problem and the EEV, i.e, VSS=RP-EEV. {\sv In Fig. \ref{VSS}, value of the stochastic solution is calculated as $\frac{RP-EEV}{EEV}*100$, `VSS' represents the series while considering uncertainties in all the parameters, and each of the other series represent the value of stochastic solution for each uncertain parameter while other parameters are replaced by their mean values. Five replications and 40 scenarios were used to obtain value of stochastic solutions. Within the parameters, dwell time has the highest impact on the accessibility to charging stations. Due to the limited capacity of charging locations, and an EV is plugged-in till the end of a driver's activity, dwell time significantly affects the accessibility to the charging stations. For the same reason, SOC's impact is minimum and contributes to a driver's decision on whether to charge or not. At lower budgets, arrival time of an EV to the community has more impact due to lesser availability of charging stations. Also, as the budget increases, due to the availability of more charging stations within the drivers' walking distance, the stochastic influence of walking has decreased. By adopting the stochastic programming approach, the overall improvement in accessibility to charging stations is 11.37 \%.}
\begin{figure}[!htbp]
	\centering
	\includegraphics[scale = 0.28]{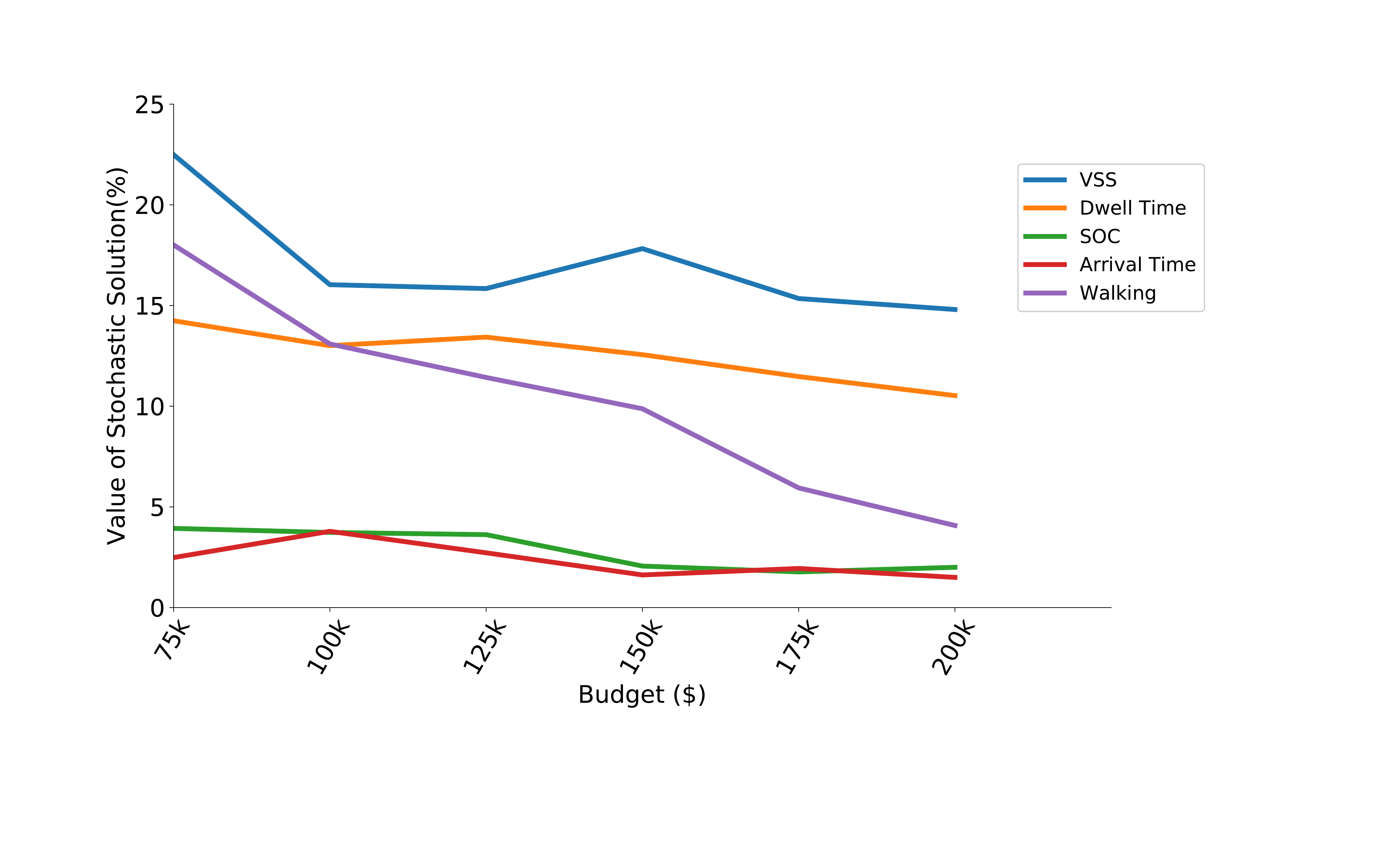}
	\vspace{-13mm}
	\captionsetup{justification=centering}
	\caption {valuating the effect of each source of uncertainty and comparing with the value of stochastic solution }
	\label{VSS}
\end{figure}
\section{Case Study and Computational Experiments}\label{Case}
We conducted a case study using data obtained from the Southeast Michigan Council of Governments (SEMCOG) and literature sources for the midtown area of Detroit, Michigan, in the US. This area includes different types of destinations, which attracts lots of traffic. There are 67 offices, 44 school-related buildings, 12 social places, 5 family-related buildings, 4 restaurants, and 3 shopping places in this area. We selected 10 parking lots as potential locations for installing {\sv chargers} and assumed that the parking lots are available for {\sv use} from 6:00 A.M to 6:00 P.M and that the capacity of each of the parking lots is based on its size. We estimated the EV demand for the case study through a two-step process. The data from SEMCOG shows that the average annual daily traffic for the Detroit midtown area is between approximately 10,000 and 14,000 vehicles and follows a uniform probability distribution. Furthermore, we calculated the EV demand for the final destination based on drivers' different activity types during the time of day and the day of the week. According to the U.S. Environmental Protection Agency's analysis, 3$\%$ and 5$\%$ of the light-duty vehicle fleet comprise EVs, and BEVs' market share can be affected by cold-temperature weather conditions \cite{pan2010locating}. Since our case study is in a cold area, we considered a 2$\%$  market share for BEVs in each case. Following a suggestion in \cite{yang2012walking}, we used a negative exponential distribution function to capture EV drivers' willingness-to-walk patterns based on the activity type, season, and community size. On average, given our parameter settings, 13$\%$  of the total demand is lost because there is no parking available within drivers' preferred walking distances.
\subsection{Scenario Generation}
We modeled uncertainties using case scenarios in the two-stage model. Each scenario represents a single day and is affected by the total number of EV drivers arriving in the community on a weekday or weekend in specific seasons of the year. Following the uniform probability distribution, each scenario occurs in each season of the year with the same probability. The arrival times of BEV drivers were estimated by Weibull distributions with parameters (8, 3) and (13, 4) for a weekend and a weekday, respectively \cite{zhong2008studying}. Based on a driver's activity, the dwell time was calculated using a Weibull distribution. The scale and shape parameters for each type of activity and type of day are provided in Table \ref{Params}. When a driver arrives in the community, a building or final destination is randomly assigned to the driver based on his/her activity type,  using a uniform distribution. As mentioned in the previous section, we use a truncated normal distribution $N(0.3, 0.1)$ with limits of 0 and 1 to estimate the SOC for an EV upon its arrival at a parking lot. This process was repeated multiple times to generate a set of scenarios.
\begin{table*}
 	\caption{Weibull distribution parameters for drivers' dwell time; Source:\cite{faridimehr2018stochastic}}
 	\centering
 	\begin{tabular}{| c | c  c  c  c  c  c |}  
 		\hline
 		Type of day & Work & Social & Family & Meal & School & Shopping \\ 
 		\hline
		Weekday & (5.89, 10) & (1.89, 10) & (1.05, 10) & (0.79, 2) & (3.61, 2) & (0.56, 2) \\
 		\hline
 		Weekend & (6.04, 6) & (2.03, 2) & (1.13, 2) & (0.79, 2) & (3.36, 10) & (0.25, 0.5) \\
 		\hline
 	\end{tabular}
 	\label{Params}
 \end{table*} 
 \subsection{Experiments and Results}
 The availability of an EVCS can offer a greater driving range for an EV and make it unnecessary to use other vehicles for longer trips. To examine the effects of different parameters and their impact on the accessibility of EVCSs in the proposed model, we studied different cases and evaluated the model with a sensitivity analysis. We considered 6:00 am -- 9:00 am, 9:00 am -- 12:00 pm, 12:00 pm -- 2:00 pm, and 2:00 pm -- 6:00 pm to be the four time slots in a day. Also, we considered $\$ $900, $ \$ $3,450, and $\$$25,000 to be the average installation costs for level 1, level 2, and level 3 {\sv chargers}, respectively{\sv \cite{smith2015costs}}. Forty scenarios were generated for the two-stage model, and 10 parking lots were used in all cases. Fig. \ref{Heat} shows the heat map for the demand distribution and the locations of parking lots. Parking lots 1 through 8 are the parking structure facilities that have the highest capacity in the area in terms of parking spots. The other two parking lots are smaller parking facilities. A darker color indicates a higher demand for a parking lot. Parking structures are considered to have a capacity for 20 stations, while the parking lots are considered to have a capacity for 5. A set of available parking lots was generated for each EV driver based on the driver's preferred walking distance. Four different metrics were used to assess the performance of the EVCS network design, including EVCS accessibility, {\sv charger} utilization, total walking distance, and average walking distance per driver. Accessibility is defined as the percentage of EV drivers who could charge their vehicles in the charging locations proposed by the two-stage model. Utilization is defined as the percentage of the total time that {\sv a charger} is used by EVs. Because the installation of public charging stations can change travelers' walking patterns, especially in an urban community, we measured the walking distance trend before and after installing charging stations.
\begin{figure}[!htbp]
	\centering
	\includegraphics[scale = 0.4]{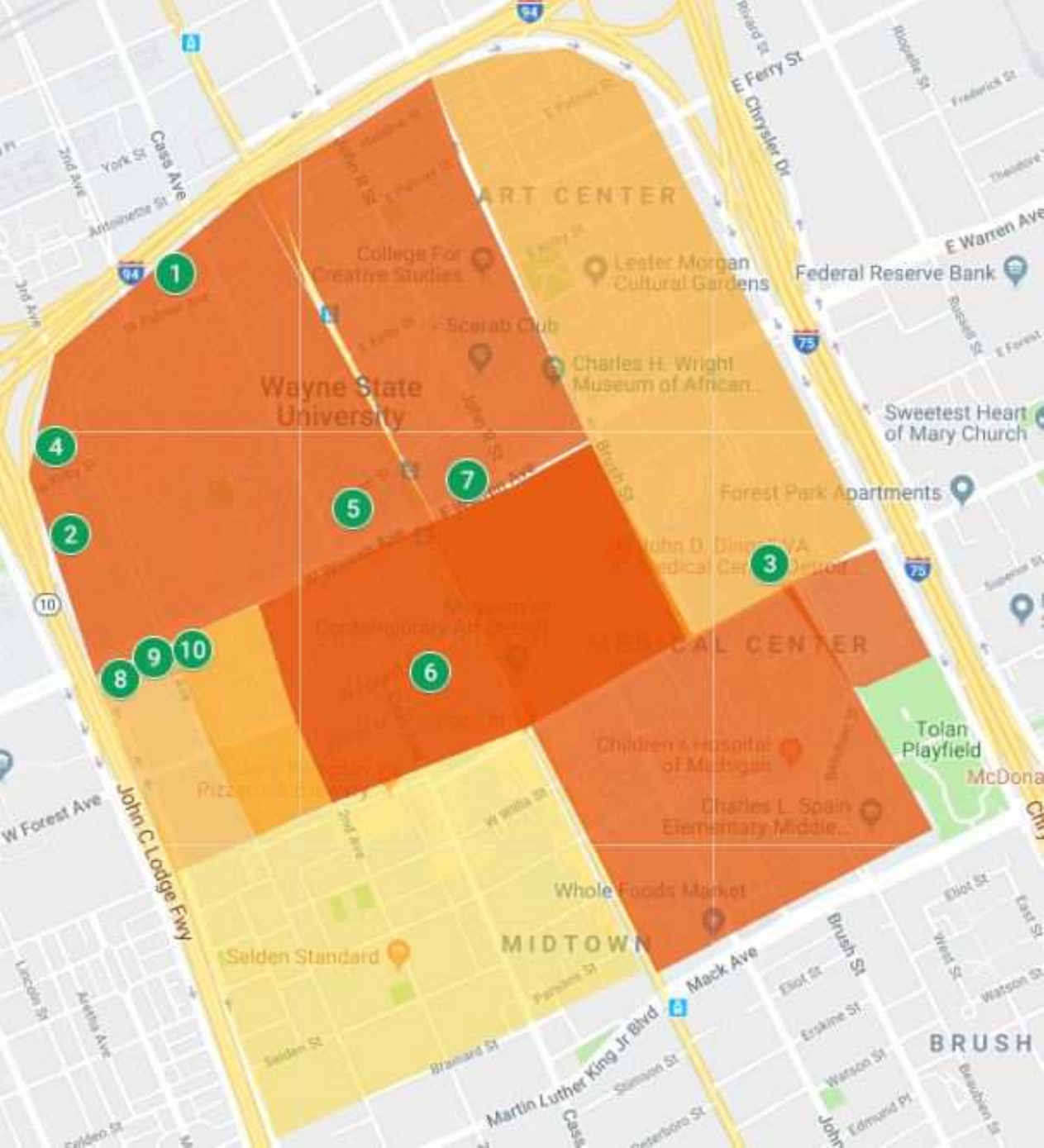}
	\captionsetup{justification=centering}
	\caption {Heat map of the demand flow and location of parking lots in the study area.}
	\label{Heat}
\end{figure}
\begin{figure}[!htbp]
	\centering
	\includegraphics[scale = 0.2]{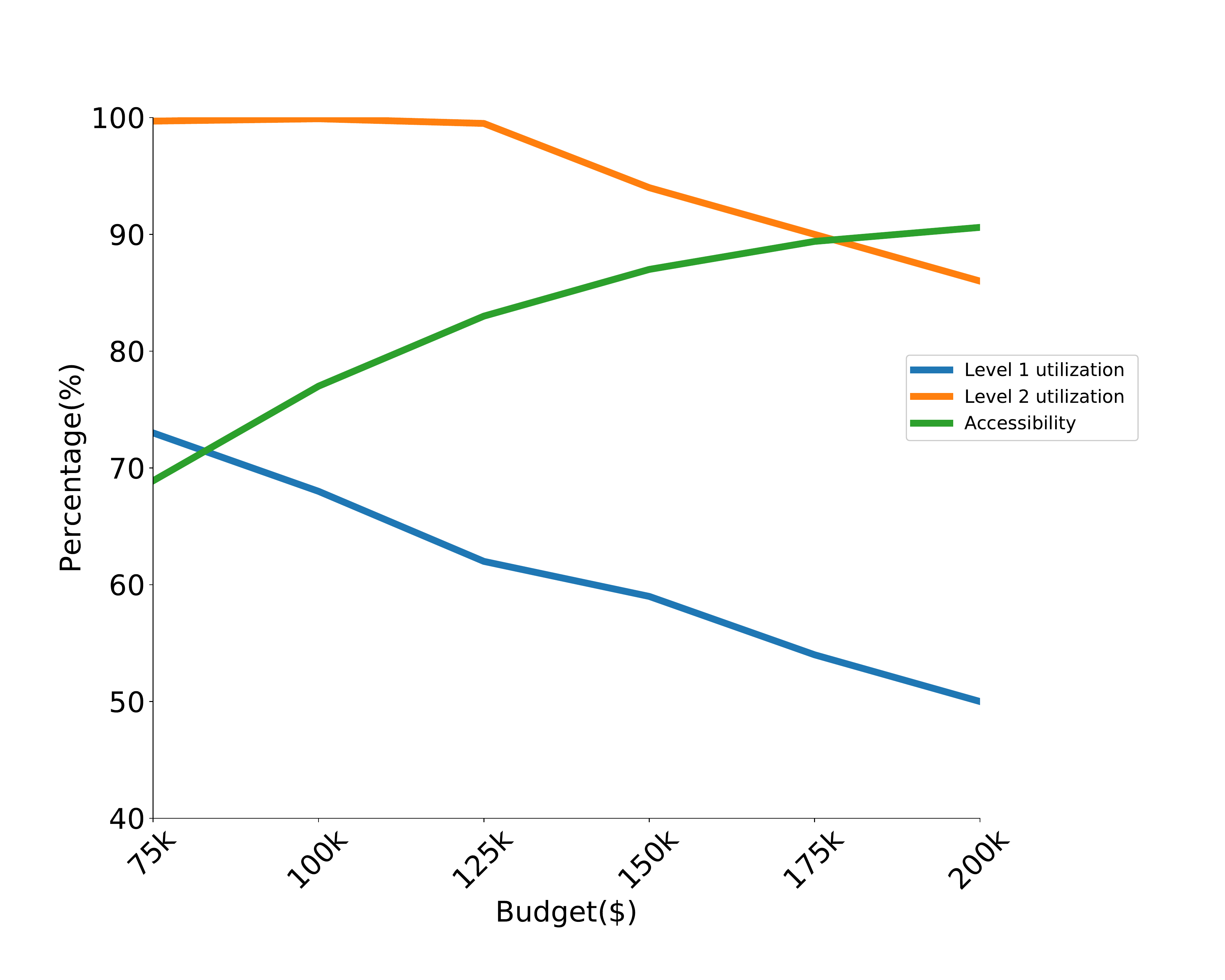}
	\captionsetup{justification=centering}
	\caption {Percentages of accessibility and charging utilization for different budget amounts.}
	\label{Access}
\end{figure}
\begin{figure}[!htbp]
	\centering
	\includegraphics[scale = 0.3]{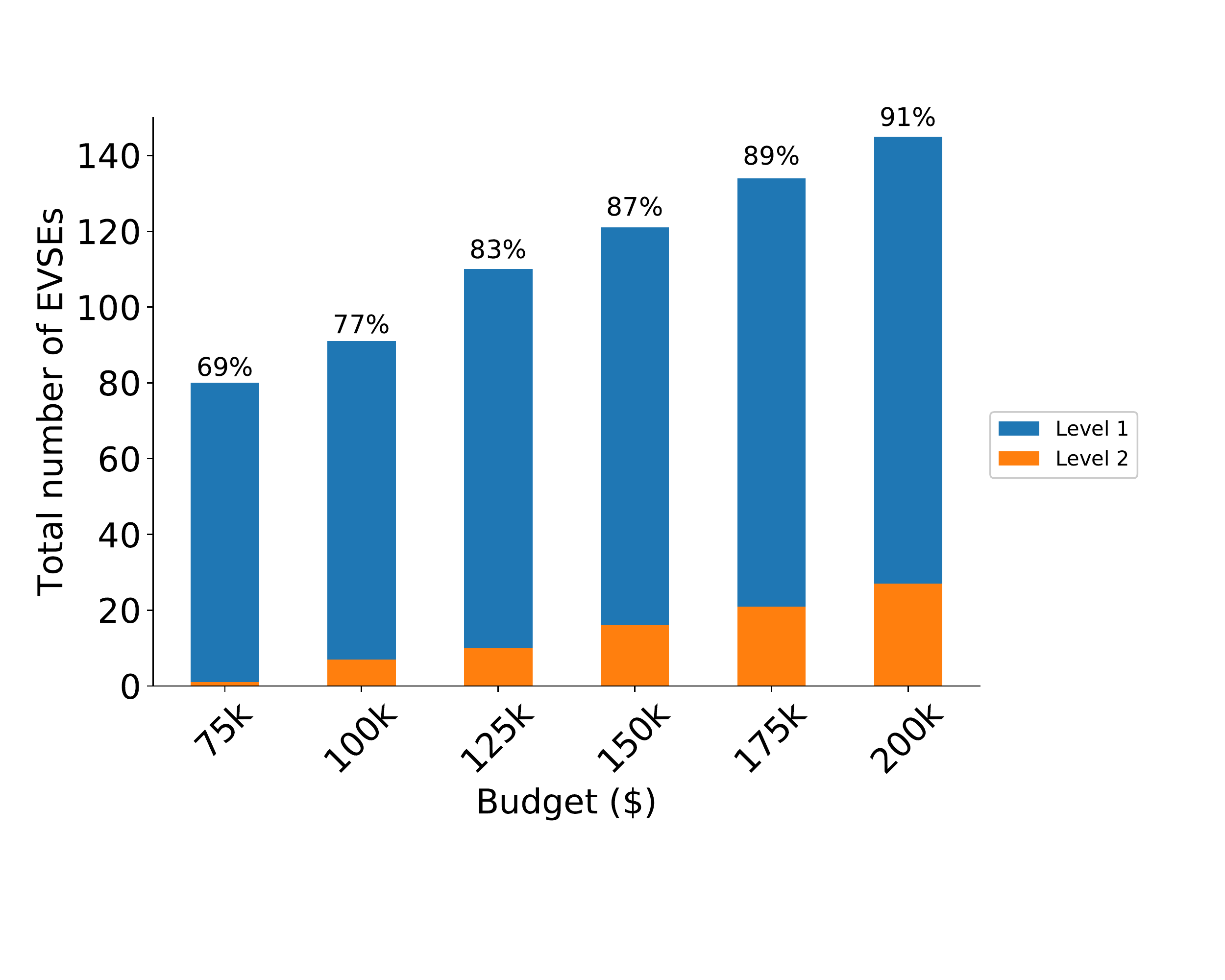}
	\vspace{-10mm}
	\captionsetup{justification=centering}
	\caption {Number of installed level 1 and level 2 chargers for different budget amounts, labeled by accessibility percentages.}
	\label{L2vsL1}
\end{figure}
Fig. \ref{Access} compares the accessibility of charging stations with the utilization of each level of {\sv chargers} as the budget increases. As expected, the results indicate that, with a budget increase, the accessibility of the charging levels also increases. In addition, level 1 utilization  decreases faster than level 2 utilization with a budget increase. This is because more level 2 chargers than level 1 chargers are installed as the  budget increases, since level 2 chargers have a higher utility for commuters. Fig. \ref{L2vsL1} compares the total number of level 1 and level 2 {\sv chargers} that are installed in parking lots based on different budgets, and these are labeled with accessibility percentages. In Fig. \ref{Parking Utilization}, the average utilization percentages for level 1 and level 2 chargers in the ten parking lots are compared for different time slots. The maximum utilization occurs between 9:00 am and 12:00 pm, and this matches the activity types of the case community, which are mostly school and work. In addition, Fig. \ref{Budget-Utilization} illustrates the trade-off between the budget size and the utilization percentage in different time slots for level 1 and level 2 chargers. Figs. \ref{Parking Utilization} and \ref{Budget-Utilization} present the utilization of chargers, which is a major factor in estimating the financial rate of return for investors.\\
\indent Although this was not the focus of the study, increases in travel options enable commuters to dedicate a part of their trip to walking or biking in order to improve their health. Thus, an optimal design of public charging infrastructures can provide opportunities for people in a community to increase their levels of physical activity. This can also improve the livability metrics within a city. {\sv Based on walking preferences, two cases were generated; pessimistic and optimistic cases. Both cases generated from a distribution given by \cite{yang2012walking}; however, in the pessimistic case the distribution is truncated for value over 0.2 mile.} Fig. \ref{Walk} compares the total walking distance and the walking distance per person among EV drivers who access public charging stations {\sv for the two cases}.
\begin{figure}[!htbp]
\subfloat[]{
	\includegraphics[angle=90,scale = 0.30]{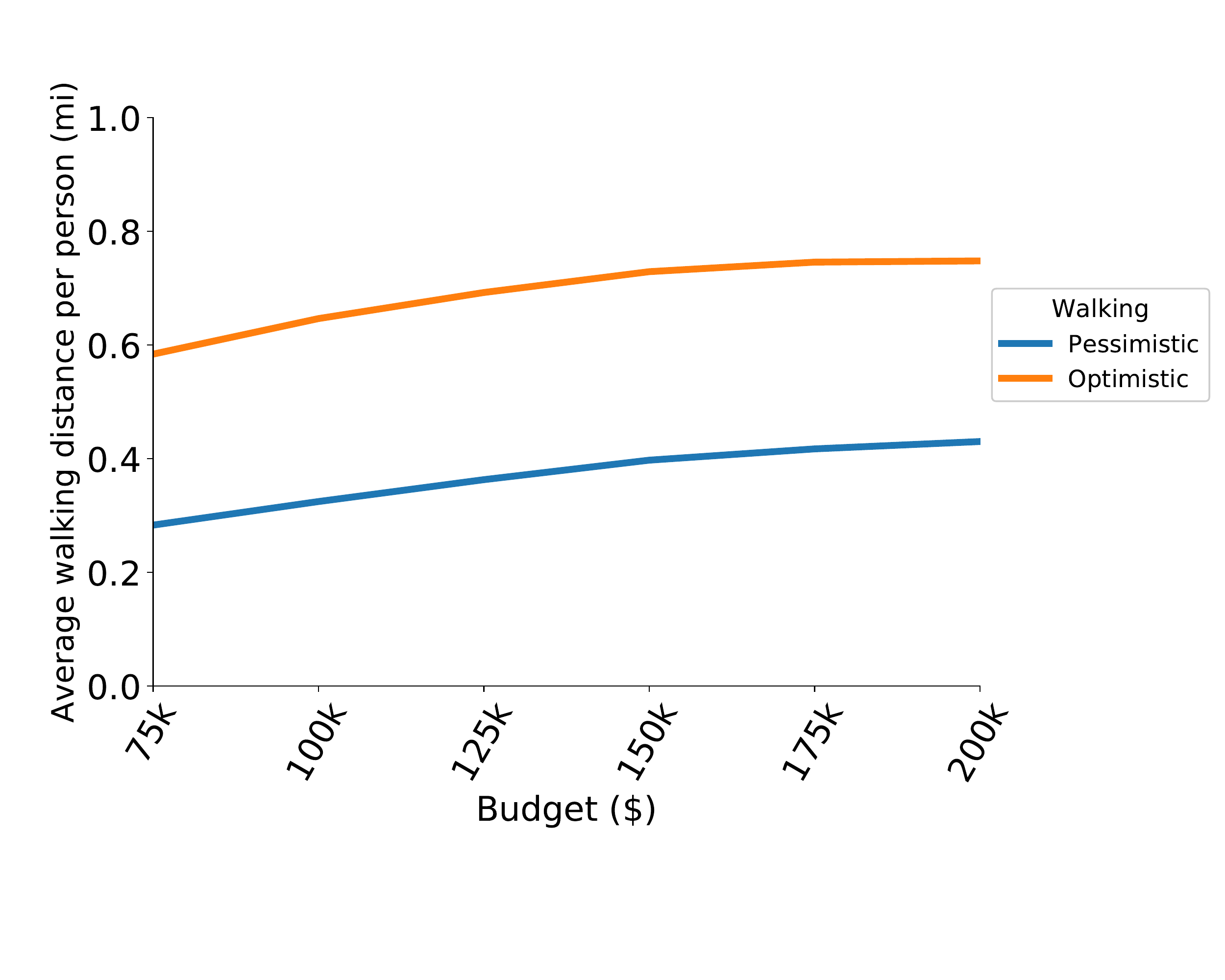}
}

\subfloat[]{
	\includegraphics[angle=90,scale = 0.30]{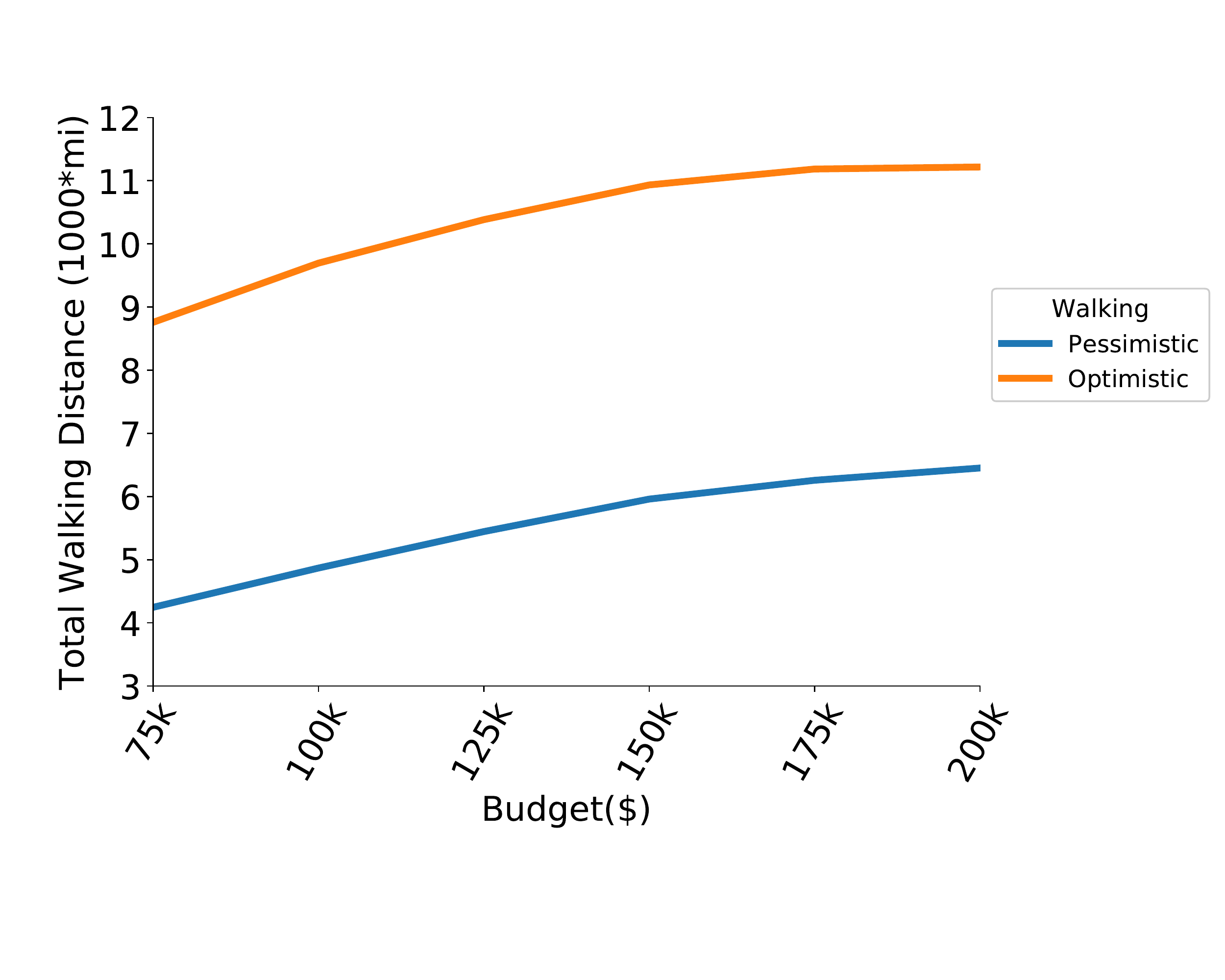}}
	\captionsetup{justification=centering}
\caption{a) Average walking distance per person and b) total walking distance for people who have access to a public EV charging station in both optimistic and pessimistic cases.}

\label{Walk}
\end{figure}
\begin{figure}[!htbp]
\subfloat[]{
	\includegraphics[angle=90,scale = 0.4]{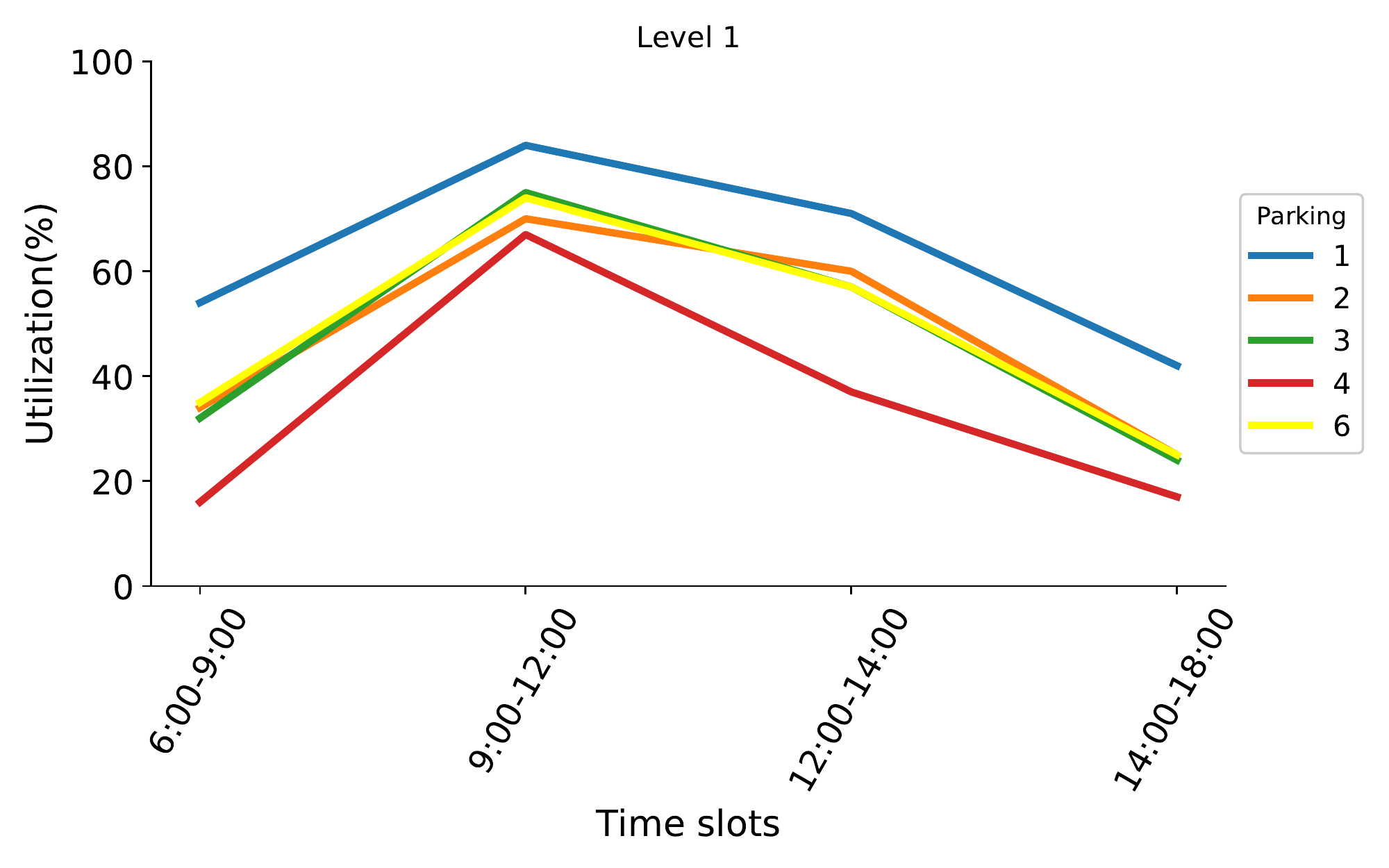}
}

\subfloat[]{
	\includegraphics[scale = 0.4]{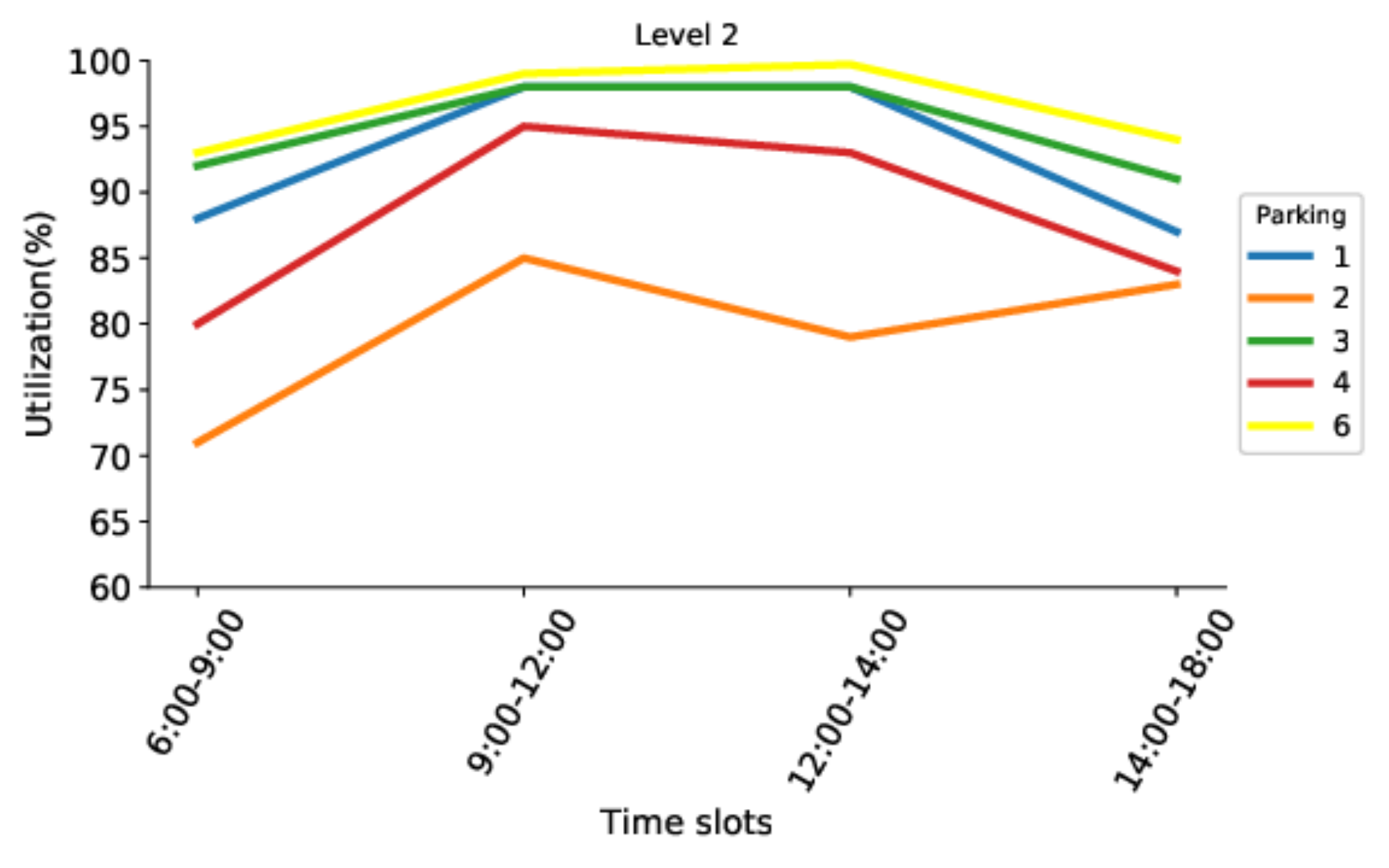}
}
\captionsetup{justification=centering}
\caption{Percentage of average utilization of a) Level 1 and b) Level 2 chargers during each time slot in five parking lots.}

\label{Parking Utilization}
\end{figure}
\begin{figure}[!htbp]
\subfloat[]{
	\includegraphics[scale = 0.4]{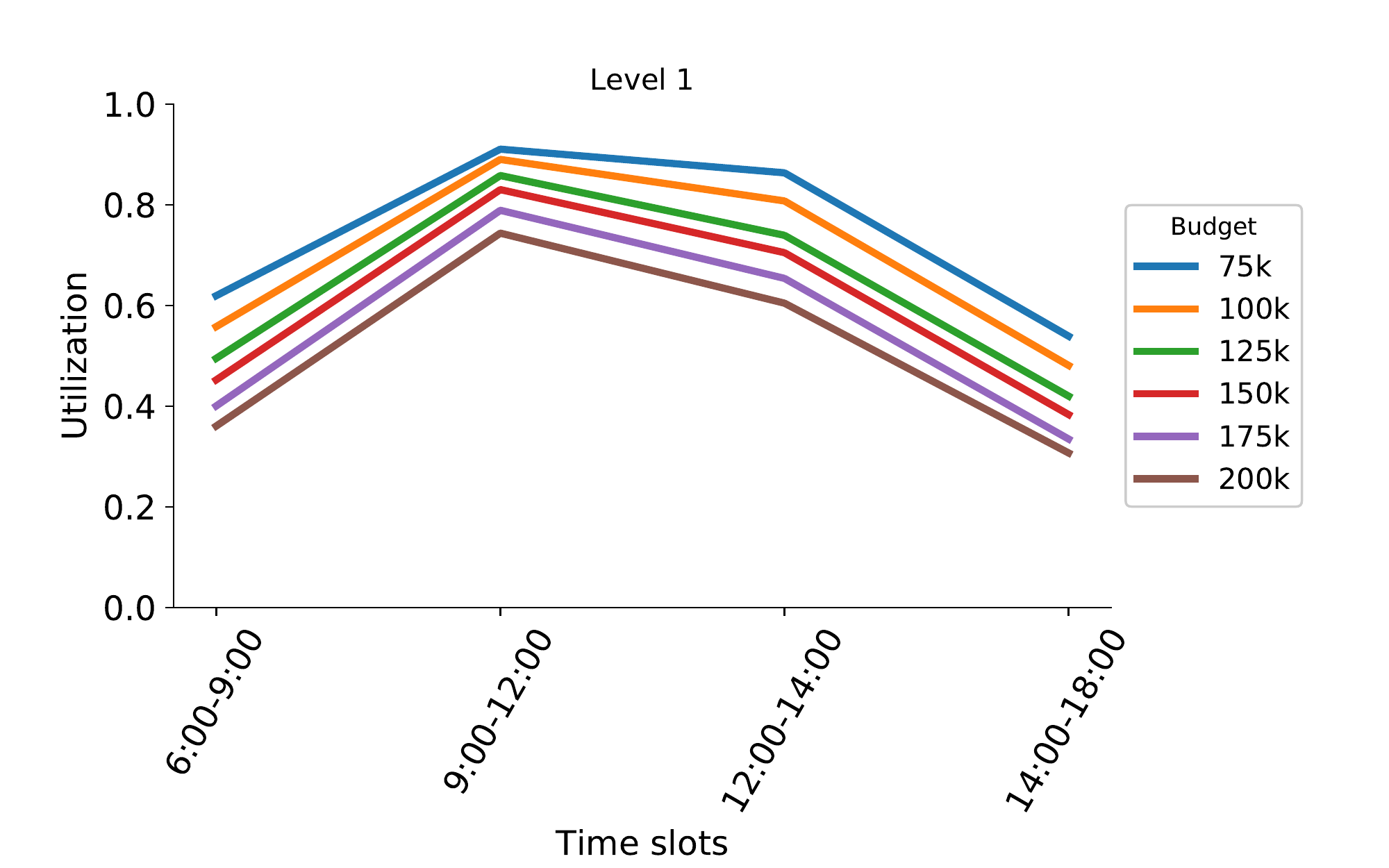}
}

\subfloat[]{
	\includegraphics[scale = 0.4]{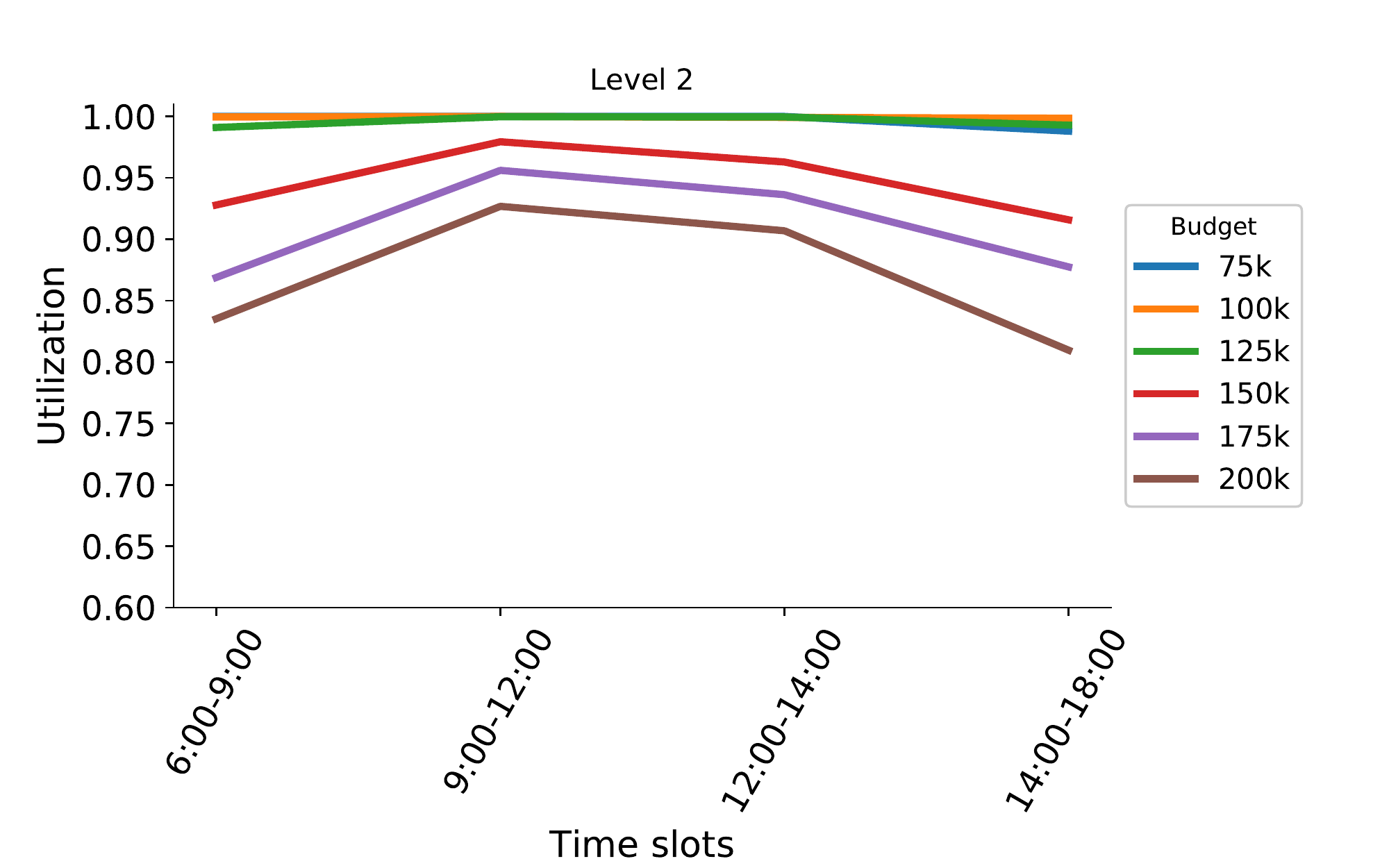}
}
\captionsetup{justification=centering}
\caption{Percentage of average utilization of a) level 1 and b) level 2 chargers in all parking lots during each time slot for different budgets.}
\label{Budget-Utilization}
\end{figure}
\\
\begin{figure*}[!htbp]\centering
\begin{multicols}{2}
    \includegraphics[width=\linewidth]{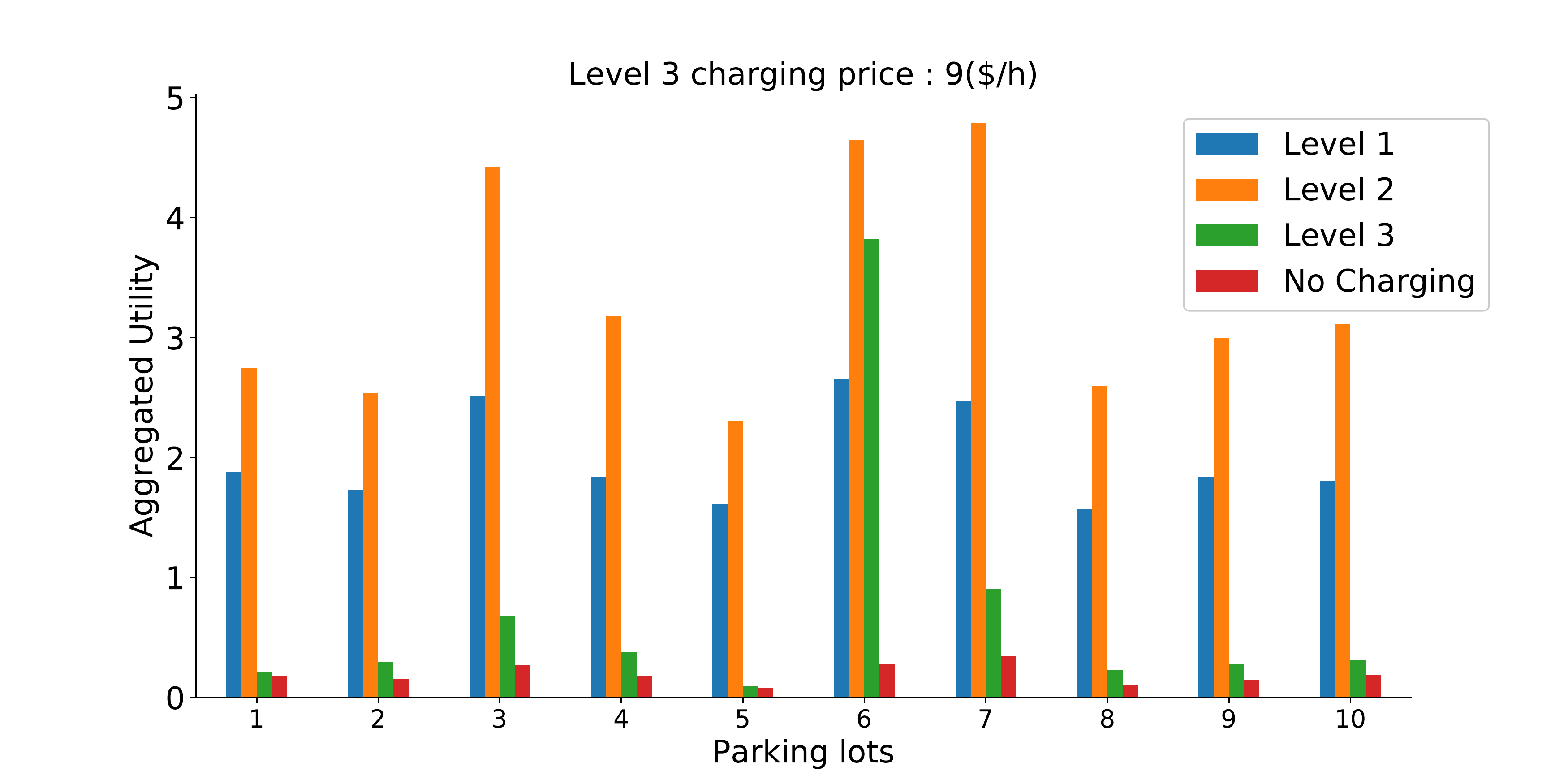}\par 
    \includegraphics[width=\linewidth]{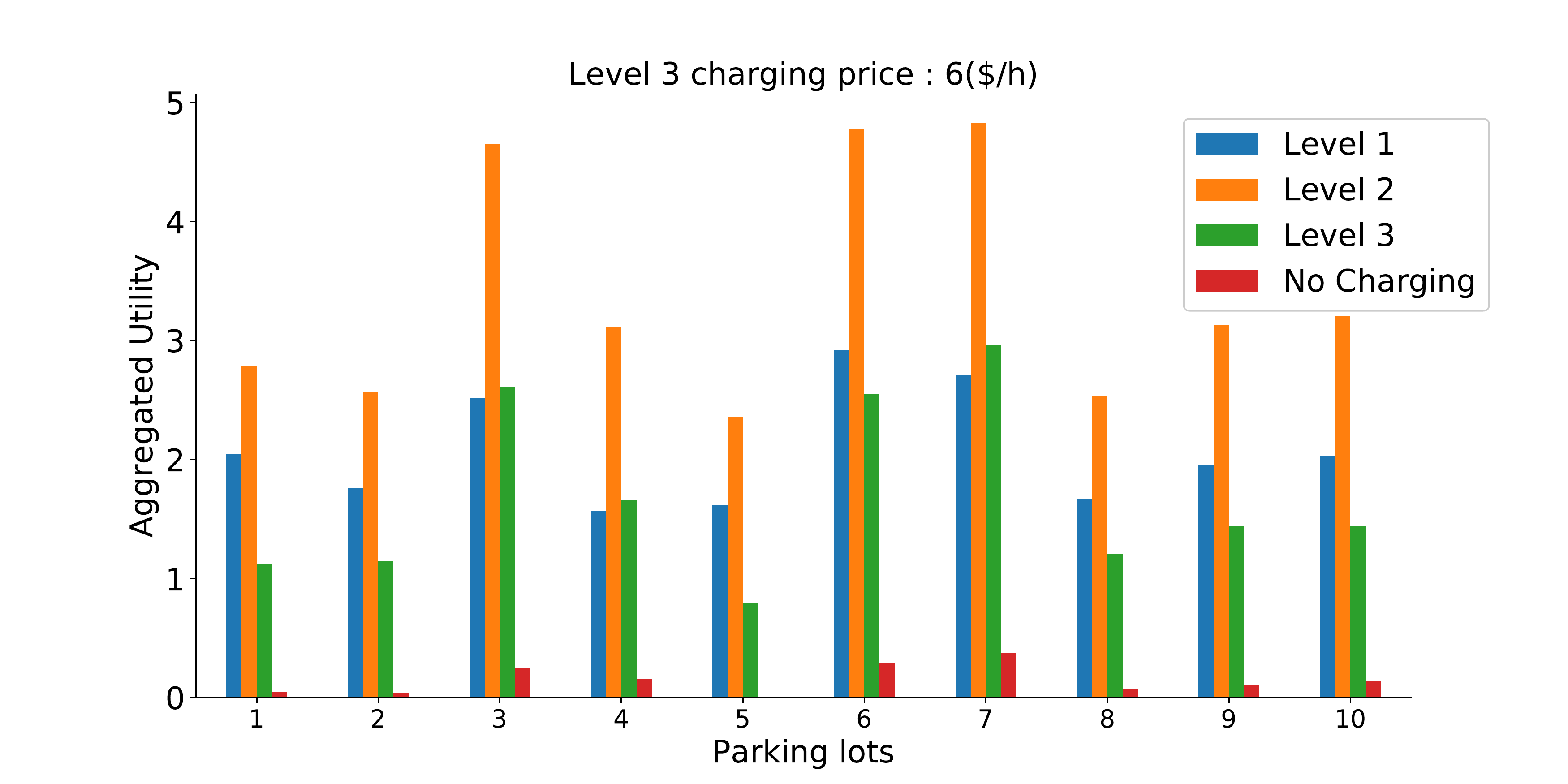}\par 
    \end{multicols}
\begin{multicols}{2}
    \includegraphics[width=\linewidth]{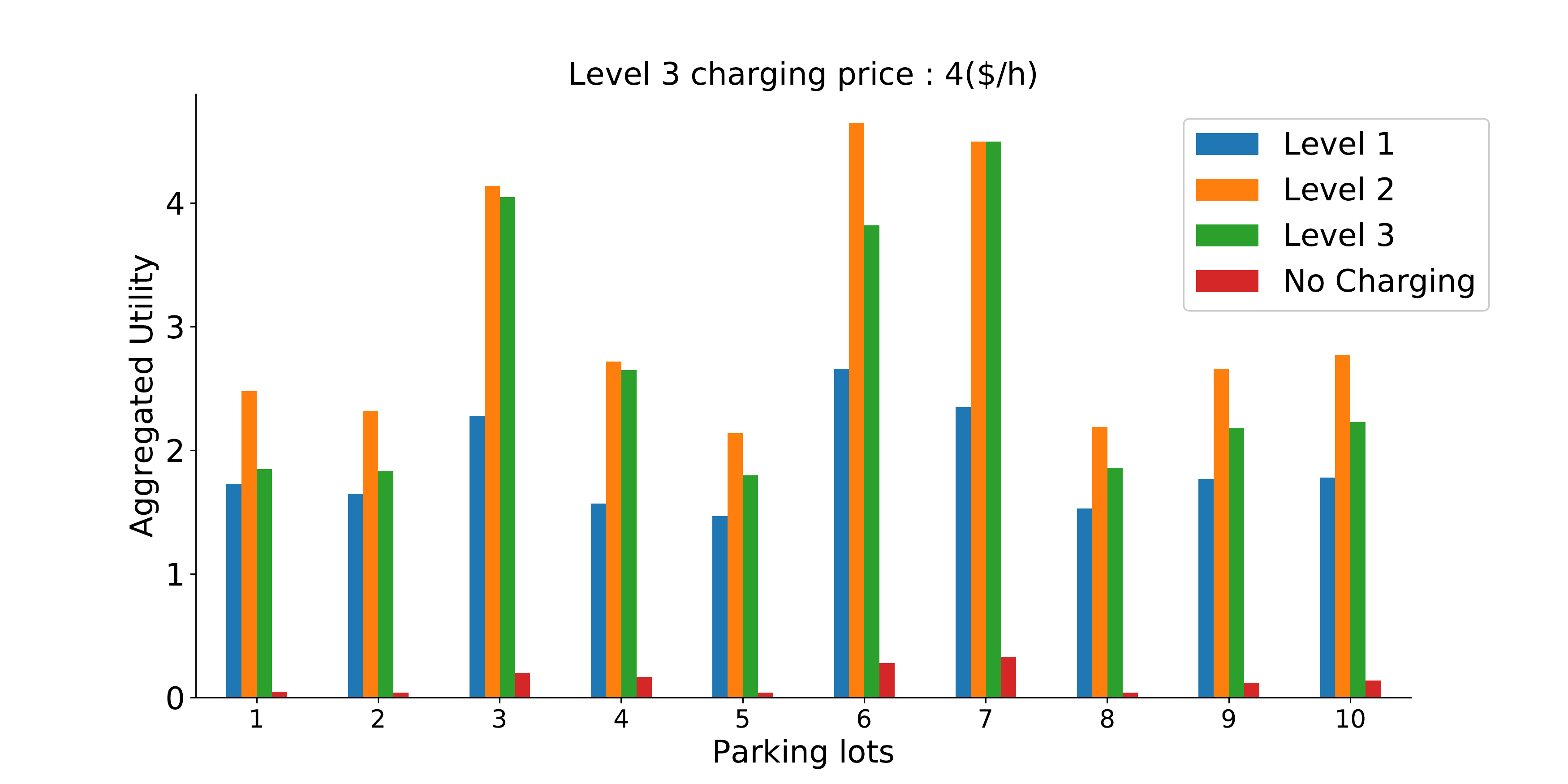}\par
    \includegraphics[width=\linewidth]{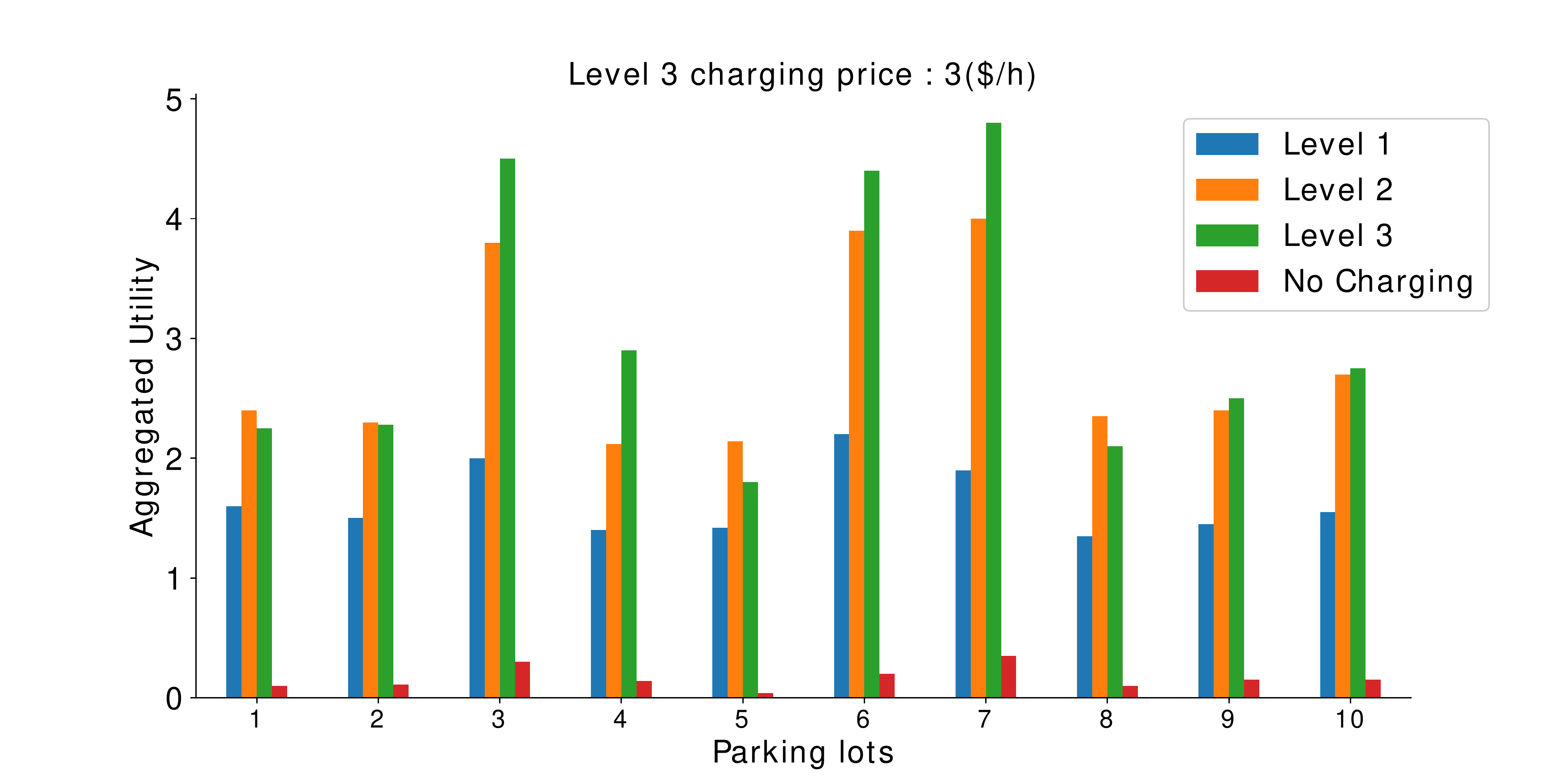}\par
\end{multicols}
\captionsetup{justification=centering}
\caption{Aggregated EV drivers' utility of using different charger levels in different parking lots when the level 3  charging price is 9(\$/h), 6(\$/h), 4(\$/h) and 3(\$/h)}
\label{price}
\end{figure*}
\begin{figure*}[!htbp]

\begin{multicols}{2}
    \includegraphics[width=1.2\linewidth]{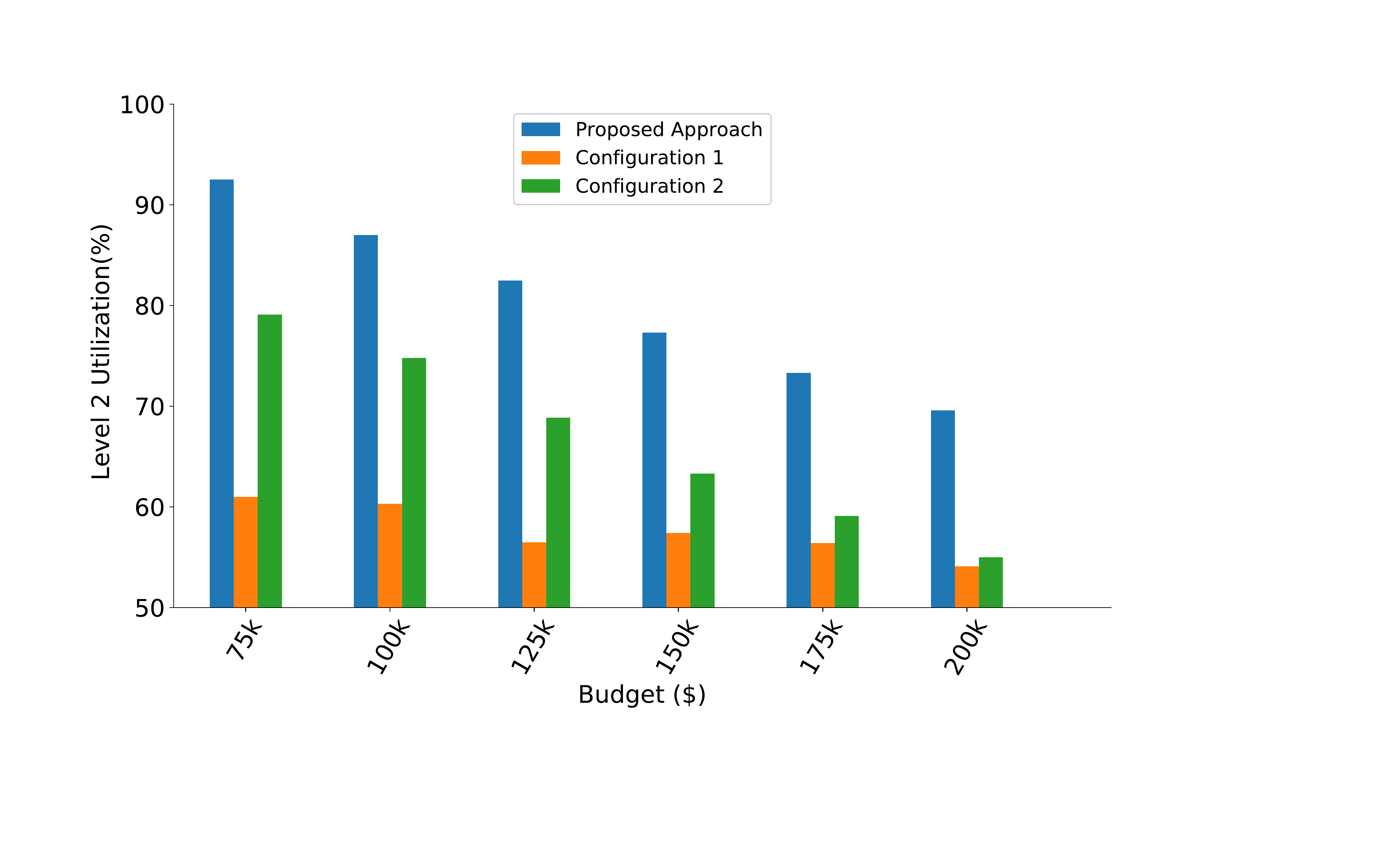}\par 
    \includegraphics[width=1.2\linewidth]{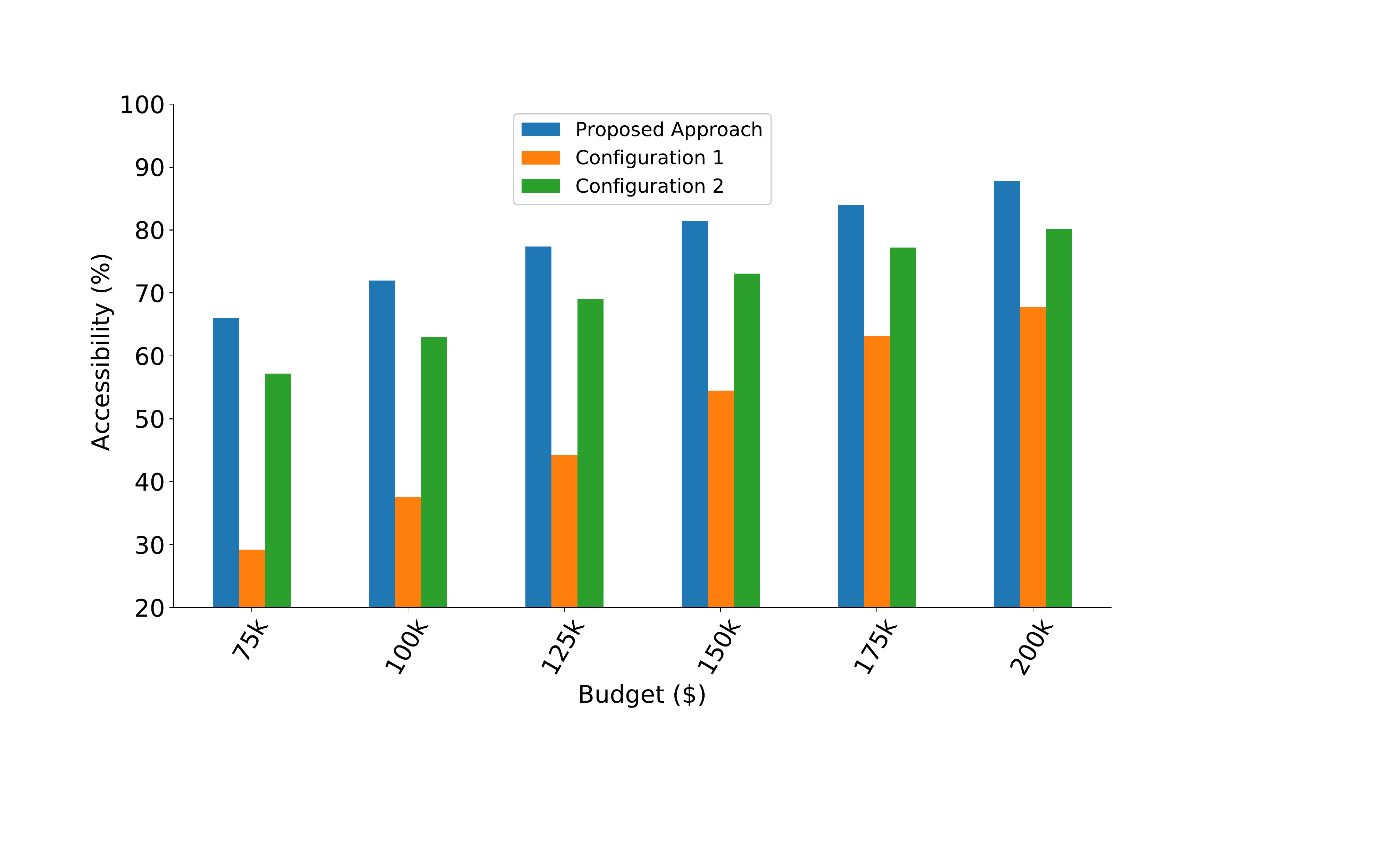}
    \end{multicols}
    \vspace*{-10mm}
    \captionsetup{justification=centering}
    \caption{Percentage of accessibility and level 2 utilization for the proposed approach and two defined configurations}
\label{sim}
\end{figure*}

As the results indicate, level 3 chargers are not installed in the parking locations. This is due to the limited budget size, since the level 3 installation cost is relatively high compared to the cost for the other levels of chargers. Based on our data sources, {\sv Fig. \ref{price}} indicates that within the urban community, people are unwilling to use level 3 chargers when the price is about 35 cents per minute. However, the utility of level 3 increases as the charging price decreases. When the charging price is finally lowered to $\$ $3 per hour, the preference for using fast chargers is higher than for level 1 and  level 2 chargers.
\subsection{Data-driven Simulation}
{\sv 
	A data driven simulation study was performed to evaluate the efficacy of the proposed work with the approach presented in \cite{faridimehr2018stochastic} which ignores  choice modeling. Two configurations were considered to establish the baseline for \cite{faridimehr2018stochastic}: configuration 1 - where all the  chargers are considered to be level 2; configuration 2 - 80\% of  the parking lots' capacity is allocated for installing level 2 chargers and the remaining capacity is assigned to level 1. Based on the network designs proposed by each of the approaches, a simulation experiment was performed to measure the `accessibility' for each driver based on the availability and choice during their arrival. If an EV driver could not find his/her first-choice of charger, the driver will search for the next best alternative. We used two performance metrics: ``accessibility" is the percentage of EV drivers who could use the chargers at their arrival; ``utilization" is calculated as the total number of hours that a charger is used by drivers over the total number of hours within the simulation period. Fig. \ref{sim} shows that the proposed approach considering choice modelling has better accessibility and utilization at each of the budget levels compared to other two configurations using the model proposed in \cite{faridimehr2018stochastic}.	For each budget, we used 200 replications for simulation, and on average, the proposed approach could increase the accessibility by 29\% and 10\%, and level 2 utilization by 23\% and 14\% compared to using configurations 1 and 2 without choice modelling, respectively.}
\section{Conclusion}\label{Con}
In this research, we propose a choice modeling approach embedded in a two-stage stochastic programming model for EV charging station network design within a community. Various factors, such as the total EV flow, arrival and dwell times, batteries' SOCs upon arrival, and the distances that EV drivers are willing to walk, are considered in the model as sources of uncertainty. Factors such as charging prices, the cost of charging at home, driving range charges, total trip distances, and dwell times are used to capture BEV drivers' charging choice behaviors. The framework suggests relationships among the budget size and the capacity and accessibility of the charging stations for EVs. A choice model utility function was helpful in determining preferences for different charger types among the EV drivers. The proposed model presents a robust {\sv charging station} network solution to any future changes in the community's pattern of willingness to walk. {\sv The computational results indicate that the optimal layout of charging stations should include a mix of different chargers. For the given data, accessibility improved with an increase in the budget, and more level 2 chargers are installed compared to the level 1. Also, with increase in budget, utilization of chargers decreased. Based on current pricing policy and utility function, level 3 is not preferred in urban communities. The simulation study in post optimization helps to study the influence of price on utility function and subsequent improvement in preferences for level 3 chargers. We ran experiments to quantify the influence of stochastic data parameters, and dwell time had the highest impact on accessibility to charging stations. Furthermore, a data-driven simulation study was conducted to evaluate the benefits of using choice modelling approach.} We solve the proposed two-stage stochastic programming models using sample average approximation and the L-shaped decomposition method. We compare the computational results with single- and multi-cut variants of the L-shaped method for a deterministic equivalent problem formulation. We present a case study using the model's results along with various insights, including (among others) a demarcation in the utility function for different charger types and the sensitivity of the optimal network to the budget. Potential research extensions might add multi-modal transportation options to the existing framework to consider the interactions between various transportation modes and EV drivers' choices. {\sv Another extension could be a study on impact of EV charging station loads on the electricity distribution network within the current framework.} From modelling perspective, other appropriate risk measures can be added to the recourse function, and the subsequent analysis can help us understand the implications of dispersion statistics while choosing optimal solutions. \textcolor{black}{From an algorithm perspective, scalability of the proposed algorithm to large planning regions needs further investigation. There are two possible avenues. Effective meta-heuristic methods can be investigated for the current formulation.  An alternate approach would to be to explore hierarchical approaches. For example, aggregate analysis can first identify the required density of charging stations for different regions/neighborhoods and detailed charging network planning can then be carried by the proposed algorithm.}
\newpage
\bibliographystyle{plain}
\bibliography{evsc}
\end{document}